\documentclass[a4paper,UKenglish,cleveref, autoref, thm-restate]{lipics-v2019}

\nolinenumbers

\usepackage{enumitem}
\usepackage{hyperref}
\usepackage{caption}
\usepackage{tikz}
\usetikzlibrary{decorations.pathreplacing}

\title{Odd Chromatic Number of Graph Classes} 

\titlerunning{Odd Chromatic Number of Graph Classes} 

\author{Rémy Belmonte}{Universit\'e Gustave Eiffel, CNRS, LIGM, F-77454 Marne-la-Vallée, France}{remy.belmonte@u-pem.fr}{0000-0001-8043-5343}{}
\author{Ararat Harutyunyan}{Universit\'{e} Paris-Dauphine, PSL University, CNRS UMR7243, LAMSADE, Paris, France}{ararat.harutyunyan@lamsade.dauphine.fr}{}{Supported by the grant from French National Research Agency under JCJC program  (DAGDigDec: ANR-21-CE48-0012)}
\author{Noleen Köhler}{Universit\'{e} Paris-Dauphine, PSL University, CNRS UMR7243, LAMSADE, Paris, France}{noleen.kohler@dauphine.psl.eu}{https://orcid.org/0000-0002-1023-6530}{Supported by the grant from French National Research Agency under JCJC program
(ASSK: ANR-18-CE40-0025-01).}
\author{Nikolaos Melissinos}{Department of Theoretical Computer Science, Faculty of Information Technology,
Czech Technical University in Prague, Czech Republic}{nik.melissinos@gmail.com}{https://orcid.org/0000-0002-0864-9803}{}

\authorrunning{R. Belmonte, A. Harutyunyan, N. Köhler and N. Melissinos} 

\Copyright{John Q. Public and Joan R. Public} 



\category{} 

\relatedversion{} 

\supplement{}

{\color{red}\acknowledgements{I want to thank \dots}}

\EventEditors{John Q. Open and Joan R. Access}
\EventNoEds{2}
\EventLongTitle{42nd Conference on Very Important Topics (CVIT 2016)}
\EventShortTitle{CVIT 2016}
\EventAcronym{CVIT}
\EventYear{2016}
\EventDate{December 24--27, 2016}
\EventLocation{Little Whinging, United Kingdom}
\EventLogo{}
\SeriesVolume{42}
\ArticleNo{23}

\newcommand{\ie}{i.e.}

\newcommand{\chiodd}{\chi_{\operatorname{odd}}}
\newcommand{\modularWidth}{\operatorname{mw}}

\renewcommand{\deg}{d}

\newtheorem{conjecture}[theorem]{Conjecture}

\begin{document}

\maketitle
\begin{abstract}
A graph is called \textit{odd} (respectively, \textit{even}) if every vertex has odd (respectively, even) degree. Gallai proved that every graph can be partitioned into two even induced subgraphs, or into an odd and an even induced subgraph. We refer to a partition into odd subgraphs as an \textit{odd colouring} of $G$. Scott [Graphs and Combinatorics, 2001] proved that a graph admits an odd colouring if and only if it has an even number of vertices. We say that a graph $G$ is $k$-odd colourable if it can be partitioned into at most $k$ odd induced subgraphs. We initiate the systematic study of odd colouring and odd chromatic number of graph classes. In particular, we consider for a number of classes whether they have bounded odd chromatic number. Our main results are that interval graphs, graphs of bounded modular-width and graphs of bounded maximum degree all have bounded odd chromatic number.
\end{abstract}

\section{Introduction}

A graph is called {\it odd} (respectively even) if all its degrees are odd (respectively even). Gallai proved the following theorem (see \cite{Lovasz93}, Problem 5.17 for a proof).
\begin{theorem}\label{thm:Gallai}
For every graph $G$, there exists:
\begin{itemize}
\item a partition $(V_1,V_2)$ of $V(G)$ such that $G[V_1]$ and $G[V_2]$ are both even;
\item a partition $(V'_1,V'_2)$ of $V(G)$ such that $G[V'_1]$ is odd and $G[V'_2]$ is even.
\end{itemize}
\end{theorem}

This theorem has two main consequences. The first one is that every graph contains an induced even subgraph with at least $|V(G)|/2$ vertices. The second is that every graph can be \textit{even coloured} with at most two colours, \ie, partitioned into two (possibly empty) sets of vertices, each of which induces an even subgraph of $G$. In both cases, it is natural to wonder whether similar results hold true when considering odd subgraphs.

The first question, known as the \textit{odd subgraph conjecture} and mentioned already by Caro~\cite{Caro94} as part of the graph theory folklore, asks whether there exists a constant $c>0$ such that every graph $G$ contains an odd subgraph with at least $|V(G)|/c$ vertices. In a recent breakthrough paper, Ferber and Krivelevich proved that the conjecture is true.

\begin{theorem}[\cite{FerberK22}]
Every graph $G$ with no isolated vertices has an odd induced subgraph of size at least $|V(G)|/10000$.
\end{theorem}

Note that the requirement that $G$ does not have isolated vertices is necessary, as those cannot be part of any odd subgraph.

The second question is whether every graph can be partitioned into a bounded number of odd induced subgraph. We refer to such a partition as an \emph{odd colouring}, and the minimum number of parts required to odd colour a given graph $G$, denoted by $\chiodd(G)$, as its \emph{odd chromatic number}. This parallels proper (vertex) colouring, where one seeks to partition the vertices of a graph into independent sets. An immediate observation is that in order to be odd colourable, a graph must have all its connected components be of even order, as an immediate consequence of the handshake lemma. Scott~\cite{Scott01} proved that this necessary condition is also sufficient. Therefore, graphs can generally be assumed to have all their connected components of even order, unless otherwise specified.

Motivated by this result, it is natural to ask how many colours are necessary to partition a graph into odd induced subgraphs. As Scott showed~\cite{Scott01}, the exist graphs with arbritrarily large odd chormatic number. On the computational side, Belmonte and Sau~\cite{BelmonteS21} proved that the problem of deciding whether a graph is $k$-odd colourable is solvable in polynomial time when $k \leq 2$, and NP-complete otherwise, similarly to the case of proper colouring. They also show that the problem can be solved in time $2^{O(k \cdot rw)} \cdot n^{O(1)}$, where $k$ is the number of colours and $rw$ is the rank-width of the input graphs. They then ask whether the problem can be solved in FPT time parameterized by rank-width alone, i.e., whether the dependency on $k$ is necessary. A positive answer would provide a stark contrast with proper colouring, for which the best algorithms run in time $n^{2^{O(rw)^2}}$ (see, e.g.,~\cite{GanianHO13}), while Fomin et al.~\cite{FominGLSZ19} proved that there is no algorithm that runs in time $n^{2^{o(rw)}}$, unless the ETH fails.\footnote{While Fomin et al. proved the lower bound for clique-width, it also holds for rank-width, since rank-width is always at most clique-width.}

On the combinatorial side, Scott showed that there exist graphs that require $\Theta(\sqrt{n})$ colours. In particular, the \textit{subdivided clique}, i.e., the graph obtained from a complete graph on $n$ vertices by subdividing\footnote{Subdividing an edge $uv$ consists in removing $uv$, adding a new vertex $w$, and making it adjacent to exactly $u$ and $v$.} every edge once requires exactly $n$ colours, as the vertices obtained by subdividing the edges force their two neighbours to be given distinct colours. More generally, and by the same argument, given any graph $G$, the graph $H$ obtained from $G$ by subdividing every edge once has $\chiodd(H) = \chi(G)$, and $H$ is odd colourable if and only if $|V(H)| = |V(G)| + |E(G)|$ is even.
Note that a subdivided clique is odd colourable if and only the subdivided complete graph $K_n$ satisfies $n\in \{k: k\equiv 0\lor k\equiv 3 \mod 4\}$. Surprisingly, Scott also showed that only a sublinear number of colours is necessary to odd colour a graph, i.e., every graph of even order $G$ has $\chiodd(G) \leq cn(\log\log n)^{-1/2}$. As Scott observed, this bound is quite weak, and he instead conjectures that the lower bound obtained from the subdivided clique is essentially tight:

\begin{conjecture}[Scott, 2001]\label{conj:Scott}
Every graph $G$ of even order has $\chiodd(G) \leq (1+o(1))c\sqrt{n}$.
\end{conjecture}

One way of  Conjecture~\ref{conj:Scott} is to consider that subdivided cliques appear to be essentially the graphs that require most colours to be odd coloured. More specifically, consider the following family of graphs $ \mathcal{B} = \{G'\text{ obtained from a graph G by adding, for every pair of vertices }$ $u,v\in V(G),\text{ a vertex }w_{uv}\text{ and edges }uw_{uv}\text{ and }vw_{uv}\text{, and }G'\text{ has even order}\}$. Note that subdivided cliques of even order are exactly those graphs in $\mathcal{B}$ where graph $G$ is edgeless, and that the graphs in $\mathcal{B}$ have $\chiodd(G') = |V(G)| \in \Theta(\sqrt{|V(G')|})$.  A question closely related to Conjecture~\ref{conj:Scott} is whether if a class of graphs $\mathcal{G}$ does not contain arbitrarily large graphs of $\mathcal{B}$ as induced subgraphs, then $\mathcal{G}$ has odd chromatic number $\mathcal{O}(\sqrt{n})$, i.e., they satisfy Conjecture~\ref{conj:Scott}. This question was already answered positively for some graph classes. In fact, the bounds provided were constant.
It was shown in~\cite{BelmonteS21} that every cograph can be odd coloured using at most three colours, and that graphs of treewidth at most $k$ can be odd coloured using at most $k+1$ colours. In fact, those results can easily be extended to all graphs admitting a join, and $H$-minor free graphs, respectively. Using a similar argument, Aashtab et al.\cite{AashAGS23} showed that planar graphs are 4-odd colourable, and this is tight due to subdivided $K_4$ being planar and 4-odd colourable, as explained above. They also proved that subcubic graphs are 4-odd colourable, which is again tight due to subdivided $K_4$, and conjecture that this result can be generalized to all graphs, i.e., $\chiodd(G) \leq \Delta + 1$, where $\Delta$ denotes the maximum degree of $G$. Observe that none of those graph classes contain arbitrarily large graphs from $\mathcal{B}$ as induced subgraphs. On the negative side, bipartite graphs and split graphs contain arbitrarily large graphs from $\mathcal{B}$, and therefore the bound of Conjecture~\ref{conj:Scott} is best possible. In fact, Scott specifically asked whether the conjecture holds for the specific case of bipartite graphs.

\textbf{Our contribution.} Motivated by these first isolated results and Conjecture~\ref{conj:Scott}, we aim to initiate the systematic study of the odd chromatic number in graph classes, and to determine which have bounded odd chromatic number. We focus on graph classes that do not contain large graphs from $\mathcal{B}$ as induced subgraphs. Our main results are that graphs of bounded maximum degree, interval graphs and graphs of bounded modular width all have bounded odd chromatic number. 

In Section~\ref{sec:degree-girth}, we prove that every graph $G$ of even order and maximum degree $\Delta$ has $\chiodd(G) \le 2\Delta -1$, extending the result of Aashtab et al. on subcubic graphs. We actually prove a more general result, which provides additional corollaries for graphs of large girth. In particular, we obtain that planar graphs of girth 11 are 3-odd colourable. We also obtain that graphs of girth at least 7 are $\mathcal{O}(\sqrt{n})$-odd colourable. While this bound is not constant, it is of particular interest since subdivided cliques have girth exactly 6.

In Section~\ref{sec:mw} we prove that every graph with all connected components of even order satisfies $\chiodd(G) \leq 3 \cdot mw(G)$, where $mw(G)$ denotes the modular-width of $G$. This significantly generalizes the cographs result from~\cite{BelmonteS21} and provides a major step towards proving that graphs of bounded rank-width have bounded odd chromatic number, which in turn would imply that the \textsc{Odd Chromatic Number} is FPT when parameterized by rank-width alone.

Finally, we prove in Section~\ref{sec:interval} that every interval graph with all components of even order is 6-odd colourable. Additionally, every proper interval graph with all components of even order is 3-odd colourable, and this bound is tight. 

We would also like to point out that all our proofs are constructive and furthermore a (not necessarily) optimal odd-colouring with the number of colours matching the upper bound can be computed in polynomial time. In particular, the proof provided in \cite{Lovasz93} of Theorem~\ref{thm:Gallai}, upon which we rely heavily is constructive, and both partitions can easily be computed in polynomial time.
An overview of known results and open cases is provided in Figure~\ref{fig:overviewGC} below.
\begin{figure}
    		\centering
            \includegraphics[width=\textwidth]{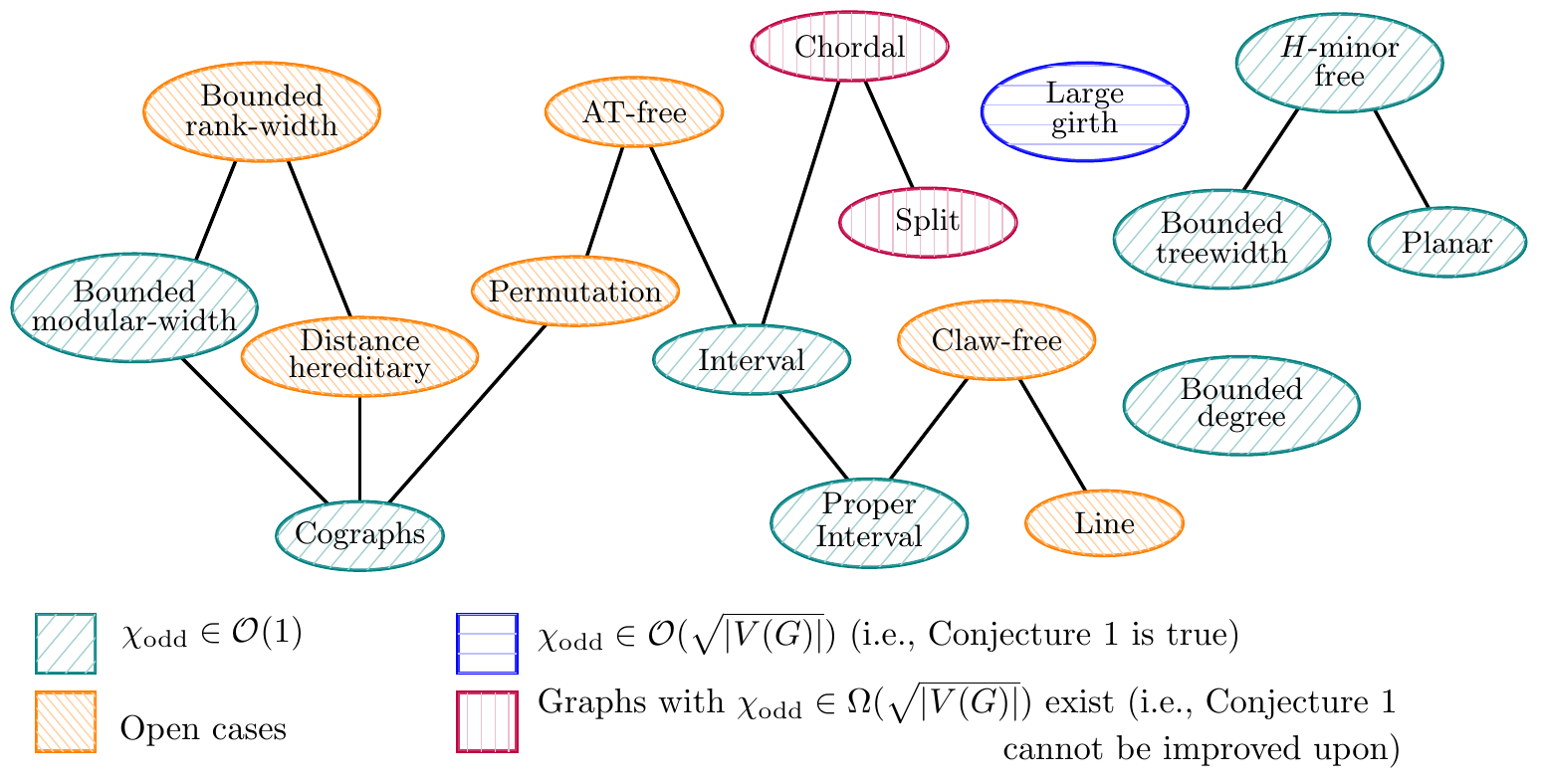}%
        \caption{Overview of known and open cases.}\label{fig:overviewGC}
\end{figure}

\section{Preliminaries}
For a positive integer $i$, we denote by $[i]$ the set containing every integer $j$ such that $1 \leq j \leq i$.
We consider a partitions of a set $X$ to be a tuple $\mathcal{P}=(P_1,\dots, P_k)$ of subsets of $X$ such that $X=\bigcup_{i\in [k]}P_i$ and $P_i\cap P_j=\emptyset$, \ie, we allow parts to be the empty set.
Let $\mathcal{P}=(P_1,\dots,P_k)$ be a partition of $X$ and $Y\subseteq X$. We let $\mathcal{P}|_Y$ be the partition of $Y$ obtained from $(P_1\cap Y,\dots, P_k\cap Y)$ by removing all empty parts. A partition $(Q_1,\dots, Q_\ell)$ of $X$ is a coarsening of a partition $(P_1,\dots, P_k)$ of $X$ if for every $P_i$ and every $Q_j$ either $P_i\cap Q_j=\emptyset$ or $P_i\cap Q_j=P_i$, \ie, every $Q_j$ is the union of $P_i$'s.

Every graph in this paper is simple, undirected, loopless and finite. We use standard graph-theoretic notation, and we refer the reader to~\cite{Diestel} for any undefined notation. For a graph $G$ we denote the set of vertices of $G$ by $V(G)$ and the edge set by $E(G)$. Let $G$ be a graph and $S \subseteq V(G)$. We denote an edge between $u$ and $v$ by $uv$. The \emph{order} of $G$ is $|V(G)|$. The \emph{degree} (resp. \emph{open neighborhood}) of a vertex $v \in V$ is denoted by $\deg_G(v)$ (resp. $N_G(v)$). We denote the subgraph induced by $S$ by $G[S]$. $G \setminus S = G[V(G) \setminus S]$. The \emph{maximum} degree of any vertex of $G$ is denoted by $\Delta$. 
We denote paths and cycles by tuples of vertices. The \textit{girth} of $G$ is the length of a shortest cycle of $G$. 
Given two vertices $u$ and $v$ lying in the same connected component of $G$, we say an edge $e$ \textit{separates} $u$ and $v$ if they lie in different connected components of $G\setminus \{e\}$. 

A graph is called odd (even, respectively) if every vertex has odd (respectively, even) degree. A partition $(V_1,\dots,V_k)$ of $V(G)$ is a $k$-\textit{odd colouring}\footnote{This definition of odd colouring is not to be confused with the one introduced by Petrusevski and Skrekovski~\cite{PetrS21}, which is a specific type of proper colouring.} of $G$ if $G[V_i]$ induces an odd subgraphs of $G$ for every $1\leq i\leq k$. We say a graph is $k$-odd colourable if it admits a $k$-odd colourable. The \textit{odd chromatic number} of $G$, denoted by $\chiodd(G)$, is the smallest integer $k$ such that $G$ is $k$-odd colourable.
The empty graph (\ie, the graph such that $V(G)=\emptyset$) is considered to be both even and odd.
Since odd colouring exists only for graphs whose every component has even size and the colouring of each component does not affect the colouring of the others, we can consider each component separately.
Therefore, it suffices to prove the statements for connected graphs of even order.

\textbf{Modular-width}
A set $S$ of vertices is called a \textit{module} if, for all $u,v\in S, N_G(u) \cap S = N(v)_G \cap S$.
A partition $\mathcal{M}= (M_1,\dots,M_k)$ of $V(G)$ is a module partition of $G$ if $\mathcal{M}$ has at least two non-empty parts and every $M_i$ is a module in $G$. 
Given two sets of vertices $X$ and $Y$, we say that $X$ and $Y$ are \textit{complete to each other} (\textit{completely non-adjacent}, respectively) if $uv\in E(G)$ ($uv \not\in E(G)$, respectively) for every $u\in X, v\in Y$.

\begin{remark}\label{rem:moduleAdjacency}
    Given any two modules $M$ and $N$ in $G$. Then either $M$ and $N$ are non-adjacent or complete to each other. 
\end{remark}
We let $G_{\mathcal{M}}$ be the module graph of $\mathcal{M}$, \ie, the graph on vertex set $\mathcal{M}$ with an edge between $M_i$ and $M_j$ if and only if  $M_i$ and $M_j$ are complete to each other (non-adjacency between modules $M_i$, $M_j$ in $G_\mathcal{M}$ corresponds to $M_i$ and $M_j$ being non adjacent in $G$). We define modular width of a graph $G$, denoted by $\modularWidth(G)$, recursively as follows. $\modularWidth(K_1)=1$, the width of a module partition $(M_1,\dots,M_k)$ of $G$ is the maximum over $k$ and $\modularWidth(G[M_i])$ for all $i\in [k]$ and $\modularWidth(G)$ is the minimum width of any module partitions of $G$.  

\section{Graphs of bounded degree and graphs of large girth}
\label{sec:degree-girth}

In this section, we study Scott's conjecture (Conjecture \ref{conj:Scott})
as well as the conjecture made by Aashtab et al.~\cite{AashAGS23}
which states that $\chiodd(G) \leq \Delta + 1$ for any graph $G$. 
We settle Conjecture \ref{conj:Scott} for graphs of girth at least 7,
and prove that $\chiodd(G) \leq 2 \Delta - 1$ for any graph $G$,
thus obtaining a weaker version of the conjecture of Aashtab et al. 
To this end, we prove the following more general theorem, which implies both of the
aforementioned results. 

\begin{theorem} \label{thm:class:bounded:chiodd}
Let $\mathcal{H}$ be a class of graphs such that:
\begin{itemize}
    \item $K_2 \in \mathcal{H}$
    \item $\mathcal{H}$ is closed under vertex deletion and 
    \item there is a $k\ge 2$ such that any connected graph $G \in \mathcal{H}$ satisfies at least one of the following properties: 
    \begin{itemize}
        \item[(I)] $G$ has two pendant vertices $u$, $v$ such that $N_{G}(u)= N_{G}(v)$ or 
        \item[(II)] $G$ has two adjacent vertices $u$, $v$ such that $d_{G}(u) +d_{G}(v) \le k $.
    \end{itemize}
\end{itemize}
Then any connected graph $G \in \mathcal{H}$ of even order has $\chiodd(G)\le k-1$.
\end{theorem}

\begin{proof}
First notice that $\mathcal{H}$ is well defined as $K_2$ has the desired properties.
The proof is by induction on the number of vertices. Let $|V(G)| = 2n$.

For $n=1$, since $G$ is connected, 
we have that $G=K_2$ which is odd. Therefore, $\chiodd(G)=1 \le k -1$ (recall that $k\ge 2$).
Let $G$ be a graph of order $2n$. 
Assume first that $G$ has two pendant vertices $u$, $v$ such that $N_{G}(u)= N_{G}(v) =\{w\}$. Then, since $G - \{u,v\}$ is connected and belongs to $\mathcal{H}$, by induction, there is an odd colouring of $G - \{u,v\}$ that uses at most $k -1$ colours. Let $(V_1,\ldots, V_{k-1})$ be a partition of $V(G)\setminus \{u,v\}$ such that $G[V_i]$ is odd of all $i \in [k-1]$. We may assume that $w  \in V_1$. 
We give a partition $V'_1,\ldots , V'_{k-1}$ of $V$ by setting $V'_1 = V_1 \cup \{u,v\}$ and $V_i'=V_i$ for all $i \in [k]\setminus \{1\}$. 
Notice that for all $i \in [k]$, $G[V'_i]$ is odd. Therefore, $\chiodd(G)\le k-1$.

Thus, we assume that $G$ has an edge $vu \in E(G)$ such that $d_{G}(v)+d_{G}(u)\leq k$. 
Note that we may assume that $k \geq 3$ for otherwise the theorem follows.
We consider two cases; $G[V(G)\setminus \{v,u\}]$ is connected and $G[V(G)\setminus \{v,u\}]$ is disconnected. 

Assume that $G[V(G)\setminus \{v,u\}]$ is connected. Since $G[V(G)\setminus \{v,u\}]$ has $|V(G)\setminus \{v,u\}| =2n-2 $ and belongs to $\mathcal{H}$, by induction, there is an odd colouring of it that uses 
at most $k-1$ colours. Let $(V_1,\ldots, V_{k-1})$ be a partition of $V\setminus \{u,v\}$, such that $G[V_i]$ is odd of all $i \in [k-1]$. 
We give a partition of $G$ into $k-1$ odd graphs as follows. Since $|N_{G}(\{u,v\})| \leq k-2$, there exists 
$\ell  \in [k-1]$ such that $V_\ell \cap N_{G}(\{u,v\}) = \emptyset$. We define a partition $U_1,\ldots, U_{k-1}$ of $V$ as follows. For all $i \in [k-1]$, if $i \neq \ell$, we define $U_i= V_i$, otherwise we set $U_{i} = V_{i} \cup \{u,v\}$. Notice that for all $i \neq \ell$, $G[U_i]$ is odd since $U_i = V_i$. Also, since 
$N_{G[U_\ell]}[v] = N_{G[U_\ell]}[u] = \{u,v\}$ and $G[V_\ell]$ is odd, we conclude that $G[U_\ell]$  is odd. Thus $\chiodd(G) \le k-1$.

Now, we consider the case where $G[V\setminus \{v,u\}]$ is disconnected. 
First, we assume that there is at least one component in $G[V\setminus \{v,u\}]$ of even order. Let $U$ be the set of vertices of this component. 
By induction,  $\chiodd (G[U]) \le k-1$ and $\chiodd (G[V\setminus U]) \le k-1$.  Furthermore, 
$|N_{G}(\{v,u\}) \cap U| \le  k - 3 $ because $G[V\setminus \{v,u\}]$ has at least two components.
Let $(U_1,\ldots, U_{k-1})$ be a partition of $U$ such that $G[U_i]$ is odd for all $i \in [k-1]$.
Also, let $(V_1,\ldots, V_{k-1})$ be a partition of $V\setminus U$ such that $G[V_i]$ is odd for all $i \in [k-1]$.
We may assume that $V_i \cap \{u,v\} = \emptyset$ for all $i \in [k-3]$. 
Since $|N_{G}(\{v,u\}) \cap U| \le  k - 3 $, there are at least two indices $l, l' \in [k-1]$ such that $U_l \cap N_{G}(\{u,v\}) = U_{l'} \cap N_{G}(\{u,v\}) = \emptyset$. We may assume that $l = k-2$ and $l' = k-1$.
We define a partition $(V_1',\ldots,V_{k-1}')$ of $V$ as follows.
For all $i \in [k-1]$ we define $V_i' = V_i \cup U_i$. We claim that $G[V'_i]$ is odd for all $i \in [k-1]$. 
To show the claim, we consider two cases; either $V'_i \cap \{u,v\} = \emptyset$ or not.
If $V'_i \cap \{u,v\} = \emptyset$, since the only vertices in $V\setminus U$ that can have neighbours in $U$ are $v$ and $u$ 
we have that $G[V'_i]$ is odd. Indeed, this holds because $N_{G}(V_i) \cap U_i = \emptyset $ and both $G[V_i]$ and $G[U_i]$ are odd.
If $V'_i \cap \{u,v\} \neq \emptyset$ then $i = k-2$ or $i = k-1$. In both cases, we know that $N_{G}(V_i) \cap U_i = \emptyset $ because the only vertices in $V\setminus U$ that may have neighbours in $U$ are $v$ and $u$ and we have assume that $u$, $v$ do not have neighbours in $U_{k-2} \cup U_{k-1}$.
So, $G[V'_i]$ is odd because $N_{G}(V_i) \cap U_i = \emptyset $ and both $G[V_i]$ and $G[U_i]$ are odd.

Thus, we can assume that all components of $G[V\setminus \{v,u\}]$ are of odd order. 
Let $\ell>0$ be the number of components, denoted by $V_1,..., V_\ell$, 
of $G[V\setminus \{v,u\}]$ and note that $\ell$ must be even. 
We consider two cases, either for all $i \in [\ell]$, one of $G[V_i \cup \{v\}]$ or $G[V_i \cup \{u\}]$ is disconnected, or there is at least one $i \in [\ell]$ such that both $G[V_i \cup \{v\}]$ and $G[V_i \cup \{u\}]$ are connected.

In the first case, for each $V_i$, $i \in [\ell]$ we call $w_i$ the vertex in $\{v,u\}$ such that $G[V_i \cup \{w_i\}]$ is connected. 
Note that $w_i$ is uniquely determined, i.e.,
only one of $v$ and $u$ can be $w_i$ for each $i \in [\ell]$. 
Now, by induction, for all $ i \in [\ell]$, $G[V_i\cup \{w_i\}]$ has $\chiodd (G[V_i\cup \{w_i\}]) \le k -1$.  
Let, for each $i \in [\ell]$, $(V_{1}^i,\ldots,V_{k-1}^{i})$ denote a partition of $V_i\cup \{w_i\}$ such that $G[V_{j}^i]$ be odd, for all $j \in [k-1]$. Furthermore, we may assume that for each $i \in [\ell]$, if $v\in V_i\cup \{w_i\}$ then $v \in V_{k-2}^{i}$. Also, we can assume that for each $i \in [\ell]$, if $u\in V_i\cup \{w_i\}$, then $u \in V_{k-1}^{i}$. 
Finally, let $I = \{i \in [\ell] \mid w_i =v \}$ and $J = \{i \in [\ell] \mid w_i = u\}$. 

We consider two cases. If $|I|$ is odd, then $|J|$ is odd since $\ell = |I|+|J|$ is even. 
Then, we claim that for the partition partition $(U_{1},\ldots,U_{k-1})$ of $V$ where 
$ U_i = \bigcup_{j \in [\ell]} V_i^j $ it holds that $G[U_i]$ is odd for all $i \in [k-1]$.  
First notice that $(U_{1},\ldots,U_{k-1})$ is indeed a partition of $V$. Indeed, the only vertices 
that may belong in more than one set are $v$ and $u$. However, $v$ belongs only to some sets $V_{k-2}^{i}$,
and hence it is no set $U_i$ except $U_{k-2}$. Similarly, $u$ belongs to no set $U_i$ except $U_{k-2}$. 
Therefore, it remains to show that $G[U_i]$ is odd for all $i \in [k-1]$. 
We will show that for any $i \in [k-1]$ and for any $x \in U_i$, $|N_{G}(x) \cap U_i|$ is odd. 
Let $x \in U_i \setminus \{v,u\}$, for some $i \in [k-1]$. Then we know that $N_{G}(x) \cap U_i = N_{G}(x) \cap V_{i}^j$ for some $j \in [\ell]$. 
Since $G[V_{i}^j]$ is odd for all $i \in [k-1]$ and $j \in [\ell]$ we have that $|N_{G}(x) \cap U_i| = 
|N_{G} (x) \cap V_{i}^j|$ is odd. 
Therefore, we only need to consider $v$ and $u$. Notice that $v \in U_{k-2} = \bigcup_{j \in [\ell]} V_{k-2}^j$ (respectively, $u \in U_{k-1} = \bigcup_{j \in [\ell]} V_{k-1}^j$). Also, $v$ (respectively, $u$) 
is included in $V_{k-2}^j$ (respectively, $V_{k-1}^j$) only if $j\in I$ (respectively, $j\in J$). 
Since $G[V_{k-2}^j]$ (respectively, $G[V_{k-1}^j]$) is odd for any $j \in [\ell]$ we have that $|N (v) \cap V_{k-2}^j|$ (respectively, $|N (u) \cap V_{k-1}^j|$) is odd for any $j \in I$ (resp. $j \in J$).
Finally, since $|I|$ and $|J|$ are odd, 
we have that $|N_{G}(v) \cap U_{k-2}| = \sum_{j \in I} |N (v) \cap V_{k-2}^j|$ and $|N_{G}(u) \cap U_{k-1}| = \sum_{j \in I} |N (u) \cap V_{k-1}^j|$ are both odd. Therefore, for any $i \in [k-1]$, $G[U_i]$ is odd and $\chiodd(G) \le k-1$.

Now, suppose that both $|I|$ and $|J|$ are even. 
We consider the partition $(U_{1},\ldots,U_{k-1})$ of $V$ where,  for all $i \in [k-3]$ 
$ U_i = \bigcup_{j \in [\ell]} V_i^j $, $U_{k-2} = \bigcup_{j \in J} V_{k-2}^j\cup \bigcup_{j \in I} V_{k-1}^j$ and $U_{k-1} = \bigcup_{j \in I} V_{k-2}^j\cup \bigcup_{j \in J} V_{k-1}^j$. We claim that for this partition it holds that $G[U_i]$ is odd for all $i \in [k-1]$. 
First notice that $(U_{1},\ldots,U_{k-1})$ is indeed a partition of $V$. Indeed, this is
clear for all vertices except for $v$ and $u$. However, $v$ only
belongs to sets of type $V_{k-2}^{i}$ for $i \in I$, and $u$ only belongs to sets of 
type $V_{k-1}^{i}$ for $i \in J$. Therefore,  $u$ or $v$ belong to no set $U_i$ except $U_{k-1}$. 
We will show that for any $i \in [k-1]$ and $x \in U_i$, $|N_{G}(x) \cap U_i|$ is odd. 
Let $x \in U_i \setminus \{v,u\}$, for some $i \in [k-1]$. Then we know 
that $N_{G}(x) \cap U_i = N_{G} (x) \cap V_{i}^j$ for some $j \in [\ell]$. 
Since $G[V_{i}^j]$ is odd for all $i \in [k-1]$ and $j \in [\ell]$ we have that $|N_{G}(x) \cap U_i| = |N_{G} (x) \cap V_{i}^j|$ is odd. 
Therefore, we only need to consider $v$ and $u$. Note that $ u,v \in U_{k-1}$. Since both $|I|$ and $|J|$ are even and $U_{k-1} = \bigcup_{j \in I} V_{k-2}^j\cup \bigcup_{j \in J} V_{k-1}^j$, we have that $|N_{G}(v) \cap U_{k-1} \setminus\{u\}|$ and $|N_{G}(u) \cap U_{k-1} \setminus\{v\}|$ are both even. Finally, since $uv \in E(G)$ we have that $|N_{G}(v) \cap U_{k-1}|$ and $|N_{G}(u) \cap U_{k-1}|$ are both odd. Hence, $\chiodd(G) \le k-1$.

Now we consider the case where there is at least one $i \in [\ell]$ where both $G[V_i \cup \{v\}]$ and $G[V_i \cup \{u\}]$ are connected. 
We define the following sets $I$ and $J$. For each $i \in [\ell]$
\begin{itemize}
    \item $i \in J$, if $G[V_i \cup \{v\}]$ is disconnected and 
    \item $i \in I$, if $G[V_i \cup \{u\}]$ is disconnected.  
\end{itemize}
Finally, for the rest of the indices, $i \in [\ell]$, which are not in $I \cup J$,
it holds that both $G[V_i \cup \{v\}]$ and $G[V_i \cup \{u\}]$ are connected.
Call this set of indices $X$ and note that by assumption $|X| \geq 1$.
It is easy to see that there is a partition of $X$ into two sets $X_1$ and $X_2$
such that both $I':=I \cup X_1$ and $J':= J \cup X_2$ have odd size.
Let $V_{I} = \bigcup_{i \in I'} V_i$ and $V_{J} = \bigcup_{i \in J'} V_i$. Now, by induction, we have that $\chiodd (G[V_I \cup \{v\}]) \le k-1$ and $\chiodd (G[V_J \cup \{u\}]) \le k-1$. 
Assume that $(V_{1}^I,\ldots,V_{k-1}^I)$ is a partition of $V_{I}$ and $(V_{1}^J,\ldots,V_{k-1}^J)$ is a partition of $V_{J}$ such that for any $i \in [k-1]$, $G[V_{i}^I]$ and $G[V_{i}^J]$ are odd. Without loss of generality, we may assume that $v \in V_{1}^I$ and $u \in V_{k-1}^J$.
Note that both $d_{G}(u)$ and $d_{G}(v)$ are at least two, which implies that $d_{G}(u) \leq k-2$ and $d_{G}(v) \leq k-2$. 
Therefore, there exists $i_0 \in [k-2]$ such that $N_{G}(v) \cap V_{i_0}^J = \emptyset$ and $j_0 \in [k-1] \setminus \{1\}$ such that $N_{G}(v) \cap V_{j_0}^I = \emptyset$. 
We reorder the sets $V_i^J$, $i \in [k-2]$, so that $i_0 = 1$ 
and we reorder the sets $V_i^I$, $i \in [k-1]\setminus \{1\}$ so that $j_0 = k-1$. 
Note that this reordering does not change the fact that $v \in V_{1}^I$ and $u \in V_{k-1}^J$. 
Consider the partition $(U_1,\ldots, U_{k-1})$ of $V$, where  
$U_i = V_{i}^I \cup V_i^J$. We claim that for all $i \in [k-1]$, $G[U_i]$ is odd.
Note that for any $x \in U_i$, we have $N_{G}(x) \cap U_i = N_{G} (x) \cap V_{i}^I$ or 
$N_{G}(x) \cap U_i = N_{G}(x) \cap V_{i}^J$. Since for any $i \in [k-1]$, $G[V_{i}^I]$ and $G[V_{i}^J]$ are odd we conclude that $G[U_i]$ is odd for any $i \in [k-1]$.   

\end{proof}

Notice that the class of graphs $G$ of maximum degree $\Delta$ satisfies the requirements of Theorem~\ref{thm:class:bounded:chiodd}. Indeed, this class is closed under vertex deletions and any connected graph in the class has least two adjacent vertices $u$, $v$ such that $d_{G}(u) + d_{G}(v) \le 2\Delta$. Therefore, the following corollary holds.

\begin{corollary}
\label{cor:bounded-degree}
Any graph $G$ of even order and maximum degree $\Delta$ has $\chiodd(G) \le 2\Delta -1$.
\end{corollary}

Next, we prove Conjecture \ref{conj:Scott} for graphs of girth at least seven.

\begin{corollary}
Any graph $G$ of girth at least $7$ has $\chiodd(G) \le \frac{3\sqrt{n}}{2}+1$ where $n = |V(G)|$.
\end{corollary}

\begin{proof}
Let $\mathcal{G}_7$ be the class of graphs of girth at least $7$. 
Note that $\mathcal{G}_7$ is closed under vertex deletion. Therefore, we need prove that 
any connected graph $G \in \mathcal{G}_7$ of even order that does not satisfy the property 
$(I)$ of the Theorem~\ref{thm:class:bounded:chiodd} 
has at least two adjacent vertices $u$, $v$ such that $d_{G}(u)+ d_{G}(v) \le 3\sqrt{n}/2+2$. 
Then, the corollary follows from the Theorem~\ref{thm:class:bounded:chiodd}

\begin{claim}
Let $G$ be a graph in $ \mathcal{G}_7$ of order $n$. If $G$ does not have two pendant vertices $u$, $v$ such that  $N_{G}(u)= N_{G}(v)$, then it has two adjacent vertices $u'$, $v'$ such that 
$d_{G}(u')+ d_{G}(v') \le \frac{3\sqrt{n}}{2}+2$. 
\end{claim}

\begin{claimproof}
Assume that for any two adjacent vertices $u$, $v$ such that $d_{G}(u)+ d_{G}(v) \ge  3\sqrt{n}/2+3$.
Let $G'$ be the graph we obtain after we remove all pendant vertices of $G$. Since each vertex of $G$ had at most one pendant vertex we have that for any edge $uv \in E(G')$, $d_{G'}(u)+ d_{G'}(v) \ge 3/2\sqrt{n}+1$. Also notice that $G'$ does not contain any pendant vertices as otherwise this vertex was attached to a pendant vertex of $G$ and this gives us an edge $uv \in E$ such that $d_{G}(u)+ d_{G}(v) =3 \le 3\sqrt{n}/2+2$. 

Let $w$ be a vertex such that $d_{G'}(w) \ge 3\sqrt{n}/4+1$. We consider all the vertices of distance at most $3$ from $w$ in  $G'$. Let $V_1$ be the set of vertices of distance one from $w$, $V_2$ be the set of vertices of distance two from $w$ and $V_3$ be the set of vertices of distance three from $w$. 
Notice that, since $G$ has girth at least $7$ we have that both $V_1$ and $V_2$
are independent sets, no two vertices in $V_1$ have a common neighbour in $V_2$,
and no two vertices in $V_2$ have a common neighbour in $V_3$.

We will compute the minimum number of vertices in these sets. 
For any $j \in \{1,2\}$, let $|V_j| = m_j$, $v_{j,i}$, $i \in [m_j]$, be the vertices of $V_j$ and $d_{j,i} = d_{G'}(v_{j,i})$ for all $i \in [m_j]$. 
For each vertex $v_{1,i}$, $i \in m_1$, select $i' \in [m_2]$ such that $v_{2,i'} \in N_{G}(v_{1,i}) \setminus \{w\}$. We note that the selected vertices $i'$ are necessarily distinct for each vertex $v_{1,i} \in V_1$.
We have $|N_{G'}[\{v_{1,i}, v_{2,i'}\}] \setminus \{w\}| = d_{1,i} +d_{2,i'} -1 \ge 3\sqrt{n}/2$. 
It follows that 
$|V(G)| \ge \sum_{i \in [m_1]} (3\sqrt{n}/2) = m_1 (3\sqrt{n}/2) \geq 9n/8 + 3 \sqrt{n}/2 > n $. This is a contradiction since $G$ has $n$ vertices. 
\end{claimproof}
\end{proof}

One may wonder if graphs of sufficiently large girth may have bounded odd chromatic number. In fact, this is far from being true, which we show in the next proposition. Recall that the \emph{chromatic number} $\chi(G)$ of a graph $G$ is the smallest integer $k$ such that $V(G)$ can be partitioned into $k$ sets each of which is independent. 

\begin{proposition} For every integer $g$ and $k$, there exist graphs of even order
and of girth at least $g$ such that $\chiodd(G) \geq k$.   
\end{proposition}

\begin{proof}
We use a classical result of Erd\H{o}s \cite{Erd1959}, which states that for all sufficiently large
$n$, there exists a $n$-vertex graph $G$ of girth at least $g$ and $\chi(G) \geq k$. Let $G$
be such a graph, with $n$ even. We may assume that $G$ has no component of odd order (otherwise, 
we can add an edge between any pair of odd components without affecting the girth or decreasing the
chromatic number). Let $H$ be the graph obtained from $G$ by subdividing each edge of $G$ once.
We claim that $\chiodd(H) \geq k$. Suppose that $\chiodd(H) \leq k-1$ and 
let $U_1,\ldots U_{k-1}$ be a partition of $V(H)$ such that $G[U_i]$ is odd for each $i \in [k-1] $.
Since $\chi(G) \geq k$, there must exist two adjacent vertices $u,v \in V(G)$ such that
both $\{u,v\} \in U_i$ for some $i \in [k-1]$. But we know that there is a vertex
$w_{uv}$ in $H$ with $N_H(w_{uv}) = \{u,v\}$. Let $U_j$ be the set containing $w_{uv}$.
Then $U_j$ is not odd, a contradiction.
\end{proof}

\begin{remark}
  In fact, by using a stronger result of Bollob\'{a}s \cite{Boll}, it is possible
 to show that for every $g$, there is $\epsilon > 0$ such that for all even
 $n$ sufficiently large, there exist connected graphs $G$ of order $n$ and girth at least
 $g$, with $\chiodd(G) > n^{\epsilon}$.
\end{remark}

Next, we obtain the following result for sparse planar graphs. 

\begin{corollary} \label{cor:planar:girth11}
Any planar graph $G$ of girth at least $11$ has $\chiodd(G)\le 3$.
\end{corollary}

\begin{proof}

Let $\mathcal{G}$ be the class of planar graphs of girth at least $11$. Notice that this class is closed under vertex deletion. 
We will show that any graph $G \in \mathcal{G}$ at least one of the following properties holds: 
\begin{itemize}
    \item[(I)] $G$ has two pendant vertices $u,v\in V(G)$ such that $N_{G}(u)=N_{G}(v)$ or
    \item[(II)] $G$ has an edge $uv \in E(G)$ such that $d_{G}(u)+ d_{G}(v) \le 4$ 
\end{itemize}

Assume that $G$ does not satisfy the property $(I)$. We construct $G'$ by deleting all pendant vertices of $G$.
If the minimum degree of $G'$ is $1$ then the property $(II)$ holds for $G$. Indeed, if $G'$ has a pendant vertex $u$ then must have a pendant vertex $v$ in $G$. Therefore, $d_{G}(u)+d_{G}(v) = 2+1 \le 4$. 

Assume that $G'$ has minimum degree $2$. Since $G'$ is also planar and has girth at least $11$ we can apply the Theorem 4.11 (Chang, Duh~\cite{ChangD18}), which states that there exists an edge $uv \in E(G')$ such that $d_{G'}(u) = d_{G'}(v)=2$. We consider two cases: either one of $u$ and $v$ were attached to a pendant vertex $v$ in $G$ or none of them were attached to a pendant vertex of $G$.
In the first case, we may assume that $u$ is attached to a pendant vertex $w$ of $G$. Then we have $d_{G}(u)+d_G(w) = 3+1\le 4$, therefore $G$ satisfies the property $(II)$. In the latter case, both $u$ and $v$ have $d_{G}(u)= d_{G}(v)=2$. Then $G$ satisfies the property $(II)$.

Now, by applying Theorem~\ref{thm:class:bounded:chiodd} to the class $\mathcal{G}$ the corollary follows.
\end{proof}

\begin{remark}
The upper bound presented in Corollary~\ref{cor:planar:girth11} is tight as  $C_{14}$ has $\chiodd (C_{14}) = 3$. 
\end{remark}
\begin{remark}
  Let $G$ be the graph obtained from $K_4$ by subdividing each edge once.
 Then $\chiodd(G) = 4$ and $G$ has girth 6. This implies that the girth condition in
 the corollary cannot be reduced below 7.  
\end{remark}

\section{Graphs of bounded modular-width}
\label{sec:mw}
In this section we consider graphs of bounded modular-width and show that we can upper bound the odd chromatic number by the modular-width of a graph.
\begin{theorem}\label{thm:oddChromaticModularWidth}
    For every graph $G$ with all components of even order $\chiodd(G)\leq 3\modularWidth(G)$.
\end{theorem}
The following is an easy consequence of \cref{thm:Gallai} which will be useful to colour modules and gain control over the parity of parts in the case the module is of even size.
\begin{remark}\label{remark:evenEvenOddColouring}
    For every non-empty graph $G$ of even order, there exists a partition $(V_1,V_2,V_3)$ of $V(G)$ with $|V_2|$, $|V_3|$ being odd such that $V[G_1]$ is odd and $G[V_2]$, $G[V_3]$ are even. This can be derived from \cref{thm:Gallai} by taking an arbitrary vertex $v\in V(G)$, setting $V_3:=\{v\}$ and then using the existence of a partition $(V_1,V_2)$ of $V(G)\setminus \{v\}$ such that $G[V_1]$ is odd and $G[V_2]$ is even.
\end{remark}
In order to prove \cref{thm:oddChromaticModularWidth} we first show how to $3$-odd colour graphs for which we have a module partition $\mathcal{M}$ such that the module graph $G_\mathcal{M}$ exhibits a particular structure (\ie{} is either a star \cref{lem:colouringStars}  or a special type of tree \cref{lem:colouringTrees}).
\begin{lemma}\label{lem:colouringStars}
    For every connected graph $G$ of even order with a module partition $\mathcal{M}= \{M_1,\dots,M_k\}$  such that $G_{\mathcal{M}}$ is a star, $\chiodd(G)\leq 3$.
\end{lemma}
\begin{proof}
    Assume that in $G_{\mathcal{M}}$ the vertices $M_2,\dots, M_k$ have degree $1$. We refer to $M_1$ as the centre and to $M_2,\dots, M_k$ as leaves of $G_\mathcal{M}$. We further assume that $|M_2|,\dots, |M_\ell|$ are odd and $|M_{\ell+1}|,\dots, |M_k|$ are even for some $\ell \in [k]$.
    We use the following two claims.
    \begin{claim}\label{claim:unionOfOddModuls}
        If $W\subseteq V(G)$ such that $G[W\cap M_i]$ is odd for every $i\in [k]$ then $G[W]$ is odd.
    \end{claim}
    \begin{claimproof}  
    First observe that the degree of any vertex $v\in W\cap M_1$ in $G[W]$ is $\deg_{G[W\cap M_1]}(v)+\sum_{i=2}^{k}|W\cap M_i|$. Since $\deg_{G[W\cap M_1]}(v)$ is odd and $|W\cap M_i|$ is even for every $i\in \{2,\dots,k\}$ (which follows from $G[W\cap M_i]$ being odd by the handshake lemma) we get that $\deg_{G[W]}(v)$ is odd.
    For every $i\in \{2,\dots, k\}$ the degree  of any vertex $v\in W\cap M_i$  in $G[W]$ is $\deg_{G[W\cap M_i]}(v)+|W\cap M_1|$ which is odd (again, because $|W\cap M_1|$ must be even).  Hence  $G[W]$ is odd.
    \end{claimproof}
    \begin{claim}\label{claim:unionOfEvenModuls}
        If $W\subseteq V(G)$ such that $G[W\cap M_i]$ is even for every $i\in [k]$, $|W\cap M_1|$ is odd and $|\big\{i\in \{2,\dots,k\}:|W\cap M_i|\text{ is odd}\big\}|$ is odd then $G[W]$ is odd.
    \end{claim}
    \begin{claimproof}
        Since $G_{\mathcal{M}}$ is a star and $M_1$ its centre we get that the degree of any vertex $v\in W\cap M_i$ for any $i\in \{2,\dots, k\}$ is $\deg_{G[W\cap M_i]}(v)+|W\cap M_1|$. Since $|W\cap M_1|$ is odd and $\deg_{G[W\cap M_i]}(v)$ is even we get that every vertex $v\in W\cap M_i$ for every $i\in \{2,\dots,k\}$ has odd degree in $G[W]$. 
        On the other hand, the degree of $v\in W\cap M_1$ is $\deg_{G[W\cap M_1]}(v)+\sum_{i=2}^{k}|W\cap M_i|$. Since $\deg_{G[W\cap M_1]}(v)$ is even and $|\big\{i\in \{2,\dots,k\}:|W\cap M_i|\text{ is odd}\big\}|$ is odd $\deg_{G[W]}(v)$ is odd. We conclude that $G[W]$ is odd.
    \end{claimproof}
    
    First consider the case that $|M_1|$ is odd. Since $G$ is of even order this implies that there must be an odd number of leaves of $G_\mathcal{M}$ of odd size and hence $\ell$ is even. Using \cref{thm:Gallai} we let $(W_1^i,W_2^i)$ be a partition of $M_i$ such that $G[W_1^i]$ is odd and $G[W_2^i]$ is even for every $i\in [k]$. Note that since $G[W_1^i]$ is odd $|W_1^i|$ has to be even and hence $|W_2^i|$ is odd if and only if $i\in [\ell]$. We define $V_1:=\bigcup_{i\in[k]}W_1^i$ and $V_2:=\bigcup_{i\in[k]}W_2^i$. Note that $(V_1,V_2)$ is a partition of $G$. Furthermore, $G[V_1]$ is odd by \cref{claim:unionOfOddModuls} and $G[V_2]$ is odd by \cref{claim:unionOfEvenModuls}. For an illustration we refer the reader to \cref{fig:colouringAStarCase1}.\\
    
    Now consider the case that $|M_1|$ is even. We first consider the special case that $\ell=1$, \ie, there is no $i\in [k]$ such that $|M_i|$ is odd. In this case we let $(W_1^i,W_2^i,W_3^i)$ be a partition of $M_i$ for $i\in \{1,2\}$ such that $G[W_1^i]$ is odd, $G[W_2^i]$, $G[W_3^i]$ are even and $|W_2^i|$, $|W_3^i|$ are odd which exists due to \cref{remark:evenEvenOddColouring}. For $i\in \{3,\dots,k\}$ we let $(W_1^i,W_2^i)$ be a partition of $M_i$ such that $G[W_1^i]$ is odd and $G[W_2^i]$ is even which exists by \cref{thm:Gallai}. We define $V_1:=\bigcup_{i\in[k]}W_1^i$, $V_2:=\bigcup_{i\in[k]}W_2^i$ and $V_3:=W_3^1\cup W_3^2$. As before we observe that $(V_1,V_2,V_3)$ is a partition of $V(G)$, $G[V_1]$ is odd by \cref{claim:unionOfOddModuls} and $G[V_2]$, $G[V_3]$ are even by \cref{claim:unionOfEvenModuls}. For an illustration see \cref{fig:colouringAStarCase2}.
    
    Lastly, consider the case that $|M_1|$ is even and $\ell>1$. By \cref{remark:evenEvenOddColouring} there is a partition $(W_1^1,W_2^1,W_3^1)$ of $M_1$  such that $G[W_1^1]$ is odd, $G[W_2^1]$, $G[W_3^1]$ are even and $|W_2^1|$, $|W_3^1|$ are odd. For $i\in \{2,\dots,k\}$ we let $(W_1^i,W_2^i)$ be a partition of $M_i$ such that $G[W_1^i]$ is odd and $G[W_2^i]$ is even which exists by \ref{thm:Gallai}. We define $V_1:=\bigcup_{i\in[k]}W_1^i$, $V_2:=W_2^1\cup\bigcup_{i=3}^k W_2^i$ and $V_3:=W_3^1\cup W_2^2$. Note that $(V_1,V_2,V_3)$ is a partition of $V(G)$. Furthermore, $G[V_1]$ is odd by \cref{claim:unionOfOddModuls} and $G[V_3]$ is odd by \ref{claim:unionOfEvenModuls}. Additionally, since $|M_1|$ is even there is an even number of $i\in \{2,\dots,k\}$ such that $|M_i|$ is odd.   Since for each  $i\in \{2,\dots,k\}$ for which $|M_i|$ is odd,  $|W_1^i|$ must be odd, we get that $|\big\{i\in \{2,\dots,k\}:|V_1\cap M_i|\text{ is odd}\big\}|$ is odd (note that $V_1\cap M_2=\emptyset$ because $W_2^2\subseteq V_3$). Hence we can use \ref{claim:unionOfEvenModuls} to conclude that $G[V_2]$ is odd. For an illustration see \cref{fig:colouringAStarCase3}.
\end{proof}

    	\begin{figure}
    		\centering
      
      \begin{minipage}[b]{.03\linewidth}
    			\centering
       \includegraphics{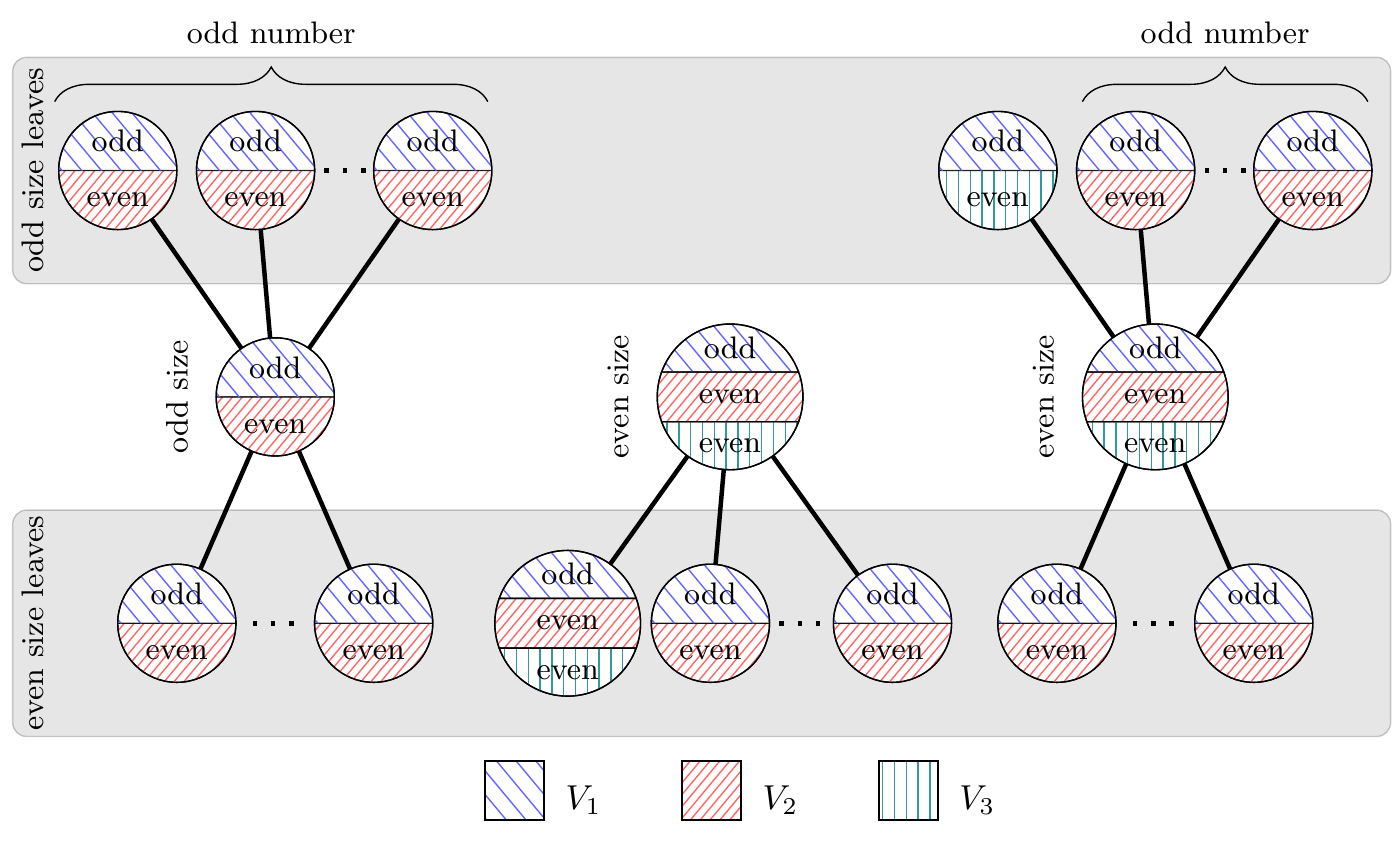}
       \vspace*{10pt}
    		\end{minipage}
            \hfill
    		\begin{minipage}[b]{.31\linewidth}
    			\centering
    				
    			\subcaption{Case: Centre is of odd size.\\
                    \phantom{bla}}\label{fig:colouringAStarCase1}
    		\end{minipage}
    		\hfill
    		\begin{minipage}[b]{.31\linewidth}
    			\centering
    			
    			\subcaption{Case: Centre is of even size and there is no odd sized leaf.}\label{fig:colouringAStarCase2}
    		\end{minipage}
            \hfill
    		\begin{minipage}[b]{.325\linewidth}
    			\centering
    			
    			\subcaption{Case: Centre is of even size and there is at least one odd sized leaf.}\label{fig:colouringAStarCase3}
    		\end{minipage}
    		\caption{Schematic illustration of the three cases in the proof of \cref{lem:colouringStars}. Depicted is the module graph $G_\mathcal{M}$ along with a partition of the modules into sets $V_1$, $V_2$ and $V_3$ such that $G[V_i]$ is odd for $i\in [3]$. }\label{fig:colouringAStar}
    	\end{figure}
     Let $G$ be a connected graph of even order with module partition $\mathcal{M}= (M_1,\dots,M_k)$ such that $G_\mathcal{M}$ is a tree. For an edge $e$  of $G_\mathcal{M}$  we let $X_e$ and $Y_e$ be the two components of the graph obtained from $G_\mathcal{M}$ by removing $e$.  We say that the tree $G_\mathcal{M}$ is colour propagating if the following properties hold.
    \begin{enumerate}[label=(\roman*)]
        \item\label{item:colProp1} There is at least one module in $\mathcal{M}$ which is not a leaf in $G_\mathcal{M}$.
        \item\label{item:colProp2} Every non-leaf module has size one.
        \item\label{item:colProp3} $|\bigcup_{M\in V(X_e)} M|$ is odd for every edge $e\in E(G_\mathcal{M})$ which is not incident to any leaf of $G_\mathcal{M}$.
    \end{enumerate} 
\begin{lemma}\label{lem:colouringTrees}
    For every connected graph $G$ of even order with a module partition $\mathcal{M}= (M_1,\dots,M_k)$ such that $G_\mathcal{M}$ is a colour propagating tree, $\chiodd(G)\leq 2$. 
\end{lemma}
\begin{proof}
    To find an odd colouring $(V_1,V_2)$ of $G$, we first let $(W_1^i,W_2^i)$ be a partition of $M_i$ such that $G[W_1^i]$ is odd and $G[W_2^i]$ is even for every $i\in [k]$. The partitions $(W_1^i,W_2^i)$ exist due to \cref{thm:Gallai}. Note that \ref{item:colProp2} implies that for every module $M_i$ which is not a leaf $|W_2^i|=1$ and $W_1^i=\emptyset$. We define $V_1:=\bigcup_{i\in[k]}W_1^i$ and $V_2:=\bigcup_{i\in [k]}W_2^i$.
    
    To argue that $(V_1,V_2)$ is an odd colouring of $G$ first consider any $v\in V(G)$ such that $v\in M_i$ for some leaf $M_i$ of $G_\mathcal{M}$. Condition \ref{item:colProp1} implies that $G_\mathcal{M}$ must have at least three vertices and hence the neighbour $M_j$  of $M_i$ cannot be a leaf due to $G_\mathcal{M}$ being a tree. Hence  $|M_j|=1$ by \ref{item:colProp2}. Hence, if $v\in W_1^i$ then $\deg_{G[V_1]}(v)=\deg_{G[W_1^i]}(v)$ since $W_1^j=\emptyset$ and therefore $\deg_{G[V_1]}(v)$ is odd. Further, if $v\in W_2^i$ then $\deg_{G[V_2]}(v)=\deg_{G[W_2^i]}(v)+1$ since $|W_2^j|=1$ and hence $\deg_{G[V_2]}(v)$ is odd. Hence the degree of any vertex $v\in M_i$ is odd in $G[V_1]$, $G[V_2]$ respectively.\\

    Now consider any vertex $v\in V(G)$ such that $M_i=\{v\}$ for some non-leaf $M_i$ of $G_\mathcal{M}$.  Let $M_{i_1},\dots, M_{i_\ell}$ be the neighbours of $M_i$ in $G_\mathcal{M}$. Let $e_j$ be the edge $M_iM_{i_j}\in E(G)$ for every $j\in [\ell]$. Without loss of generality, assume that $M_i\notin V(X_{e_j})$  for every $j\in [\ell]$. 
    By \ref{item:colProp3} we have that $|\bigcup_{M\in V(X_{e_j})} M|$ is odd whenever  $M_{i_j}$ is not a leaf in $G_\mathcal{M}$. Hence $|X_{e_j}|\equiv |M_{e_j}|\mod 2$ for every $j\in [\ell]$ for which $M_{i_j}$ is not a leaf in $G_\mathcal{M}$. On the other hand, as a consequence of the handshake lemma we get that $|W_2^{i_j}|$ is odd if and only if $|M_{i_j}|$ is odd.  
    Hence the following holds for the parity of the degree of $v$ in $G[V_2]$.
    $$\deg_{G[V_2]}(v)=|\{j\in[m]: \deg_{G_\mathcal{M}}(M_{i_j})=1\}|+\bigcup_{j\in[m] \atop \deg_{G_\mathcal{M}}(M_{i_j})\geq 2 } |W_2^{i_j}|\equiv |V(G)\setminus M_i|\mod 2.$$
     Since $G$ has even order we conclude that $\deg_{G[V_2]}(v)$ is odd and hence $(V_1,V_2)$ is an odd colouring of $G$.     
\end{proof}

We now show that given a graph $G$ with module partition $\mathcal{M}$ we can decompose the graph in such a way that the module graph of any part of the decomposition with respect to the module partition $\mathcal{M}$ restricted to the part of the decomposition is either a star or a colour propagating tree. To obtain the decomposition we use a spanning tree $G_\mathcal{M}$ and inductively find a non-separating star or a colour propagating tree. In order to handle parity during this process we might separate a module into two parts of the decomposition. 
\begin{lemma}\label{lem:decompositionOfModuleGraph}
    For every graph $G$ of even order and module partition $\mathcal{M}=(M_1,\dots,M_k)$ there is a  partition $\hat{\mathcal{M}}$ of $V(G)$ with at most $2k$ many parts such that there is a coarsening $\mathcal{P}$ of $\hat{\mathcal{M}}$ with the following properties. $|P|$ is even for every part $P$ of $\mathcal{P}$. Furthermore, for every part $P$ of $ \mathcal{P}$ we have that $\hat{\mathcal{M}}|_P$ is a module partition of $G[P]$ and $G[P]_{\hat{\mathcal{M}}|_P}$ is either a star (with at least two vertices) or a colour propagating tree.
\end{lemma}
\begin{proof}
    We use the following extensively throughout the proof.
    \begin{claim}\label{claim:restrictionOfModulePartitions}
        If $\mathcal{N}$ is a module partition of a graph $H$ and $W\subseteq V(G)$ such that $\mathcal{N}|_W$ has  at least two parts, then $\mathcal{N}|_W$ is a module partition of $G[W]$. 
    \end{claim}
    \begin{claimproof}
        Assume that this is not the case and there is a part $N$ of $\mathcal{N}|_W$ which is not a module in $G[W]$. By construction there is a part $N'$ of $\mathcal{N}$ such that $N\subseteq N'$. Since $N$ is not a module in $G[W]$  there are vertices $u,v\in N$, $w\in W\setminus N$ such that $uw\in E(G[W])$ and $vw\notin E(G[W])$. Since $N\subseteq N'$ this implies that $N'$ cannot be a module in $G$, a contradiction. 
    \end{claimproof}
    We use an induction on the number of modules in $\mathcal{M}$ to find  partitions $\hat{\mathcal{M}}$ and $\mathcal{P}$.
    Observe that in case $G_{\mathcal{M}}$ is a star or a colour propagating tree we can set $\hat{\mathcal{M}}:=\mathcal{M}$ and $\mathcal{P}:=(V(G))$ which satisfies the conditions of the statement. Hence assume that $G_{\mathcal{M}}$ is neither a star nor a colour propagating tree.  We use the two following claims to conduct our inductive argument. 
    \begin{claim}\label{claim:inductionStepCase1}
        Let $H$ be a graph with module partition $\mathcal{N}=(N_1,\dots,N_\ell)$ such that $H_\mathcal{N}$ is neither a star nor a colour propagating tree.
        If $H_\mathcal{N}$ is a tree  
        then  there is a partition $\hat{\mathcal{N}}=(\hat{N}_1,\dots,\hat{N}_{\hat{\ell}})$ of $V(G)$ with $\hat{\ell}\leq \ell+1$ and a coarsening $\mathcal{Q}=(Q_1,Q_2)$ of $\hat{\mathcal{N}}$ with the following properties. $|Q_i|$ is even and $\hat{\mathcal{N}}|_{Q_i}$ is a module partition of $H[Q_i]$ for $i\in[2]$. Furthermore, $H[Q_1]$ is connected   and  $H[Q_2]_{\hat{\mathcal{N}}|_{Q_2}}$ is either a star  or a colour propagating tree. Additionally, for any fixed index $i\in [\ell]$  we can enforce that  $N_{i}\cap Q_1\not= \emptyset$. 
    \end{claim}
    \begin{claimproof}
    First observe that since $H_\mathcal{N}$ is a tree but neither a colour propagating tree nor a star we know that either
    \begin{itemize}
            \item there is a non-leaf vertex $N$ in $H_\mathcal{N}$ with $|N|>1$  which has at least one non-leaf neighbour or 
            \item there is an edge  $e\in E(H_{\mathcal{N}}) $  not incident to any leaf of $H_{\mathcal{N}}$ and $|\bigcup_{N\in V(X_e)} N|$ is even. 
        \end{itemize}  
        Assume we have fixed $i\in [\ell]$ ($N_{i}$  will be the part which is guaranteed to be partially contained in $Q_1$). For any part $N\not= N_i$ of $\mathcal{N}$ we let  $e_N$ be an edge incident to $N$ which  separates $N$ from   $N_{i}$. Let $X_N$ be the component of $H_\mathcal{N}$ after removing $e_N$ which contains $N$.  Let $Z$ be the set of non-leaf parts $N\not= N_i$ in $H_\mathcal{N}$ such that either  $|N|>1$ or $e_N$ is not incident to a leaf and $X_N$ is of even order. Note that by our previous observation we know that $Z\cup \{N_i\}$ cannot be empty.  Finally, we in case $Z\not= \emptyset$ we let $N\in Z$ be a part with minimal $|V(X_N)|$ amongst all parts in $Z$. In case $Z=\emptyset$ we define $N:= N_i$. In this case we let $e_N$ be an edge incident to $N_i$ and some other non-leaf vertex and $X_N$ the component containing $N_i$ after removing $e_N$. Observe that in case $Z=\emptyset$ we get as an immediate consequence that $|N_i|>1$ and $|X_N|$ is odd. 
        Now observe that in any case our  choice of $N$ guarantees that $|N'|=1$ for every non-leaf part $N'\not= N$ of $X_N$ and  for every edge $e\in E(X_N)$ not incident to a leaf of $X_N$ we have that  $|\bigcup_{N'\in V(X_e)} N'|$ is odd. Furthermore, since $N$ is not a leaf $X_N$ has more than one vertex. 

        First consider the case that $|N|=1$.  We set $\hat{\mathcal{N}}:=\mathcal{N}$, $Q_2:=\bigcup_{N'\in V(X_N)}N'$ and $Q_1:=V(H)\setminus Q_2$. Since $N\in Z$ and $|N|=1$ we know that $|Q_1|$  and $|Q_2|$ are even. Furthermore, \cref{claim:restrictionOfModulePartitions} implies that $\mathcal{N}|_{Q_i}$ is a module partition of $H[Q_i]$ for $i\in [2]$. By construction $H[Q_1]$ is connected and $H[Q_2]_{\hat{\mathcal{N}}|_{Q_2}}$ is a colour propagating tree. Lastly, observe that by choosing $e_N$ to be an edge separating $N$ from  $N_{i}$ (since $|N|=1$ we get $N\not= N_i$) we get  that  $N_{i}\subseteq Q_1$ as required. \\
        
         Now consider the case that $|N|> 1$. First assume that all neighbours of $N$ in $X_N$ are leaves. In this case let $N'\subseteq N$ such that $N'\not= \emptyset$ and $|N'\cup \bigcup_{N''\in V(X_N)}N''|$ is even. Now we define 
         $\hat{\mathcal{N}}$ to be the partition obtained from $\mathcal{N}$ by removing part $N$ and adding $N'$ and $N\setminus N'$. We further let $Q_2:=N'\cup \bigcup_{N''\in V(X_N)}N''$ and $Q_1:=V(H)\setminus Q_2$. By construction $|Q_1|$, $|Q_2|$ are even. Furthermore,  $\hat{\mathcal{N}}|_{Q_1}$ must contain at least two parts since $H_\mathcal{N}$ is not a star.  Hence $\hat{\mathcal{N}}|_{Q_i}$ is a module partition of $H[Q_i]$ for $i\in [2]$ by \cref{claim:restrictionOfModulePartitions}. Furthermore, $H(Q_1)$ is connected and $H[Q_2]_{\hat{\mathcal{N}}|_{Q_2}}$ is a star. In case $N=N_i$ recall that $X_N$ is of odd order and hence we can pick $N'$ such that $N\setminus N'\not= \emptyset$ which implies $N_i\cap Q_1\not= \emptyset$. Finally, in case $N\not= N_i$ we get $N_i\subseteq Q_1$ as in the previous case. 

         On the other hand,  assume that $N$ has at least one non-leaf neighbour $N'$ in $X_N$. Choose an arbitrary vertex $n\in N$. We define  $\hat{\mathcal{N}}$ to be the partition obtained from $\mathcal{N}$ by removing $N$ and adding $\{n\}$ and $N\setminus \{n\}$. We additionally set $Q_2:=\{n\}\cup \bigcup_{N''\in V(X_{N'})}N''$ and $Q_1:=V(H)\setminus Q_2$. Note that since $NN'$ is an edge between non-leaf vertices we get that $|\bigcup_{N''\in V(X_{N'})} N''|$ is odd and hence $|Q_1|$ and $|Q_2|$ must be even. Since $N\setminus \{n\}\not=\emptyset$ we get that $\mathcal{N}|_{Q_1}$ contains at least two parts and hence $\hat{\mathcal{N}}|_{Q_i}$ is a module partition of $H[Q_i]$ for $i\in [2]$ by \ref{claim:restrictionOfModulePartitions}. Additionally, since $N\setminus \{n\}\not=\emptyset$  we have that $H(Q_1)$ must be connected. Finally,  $H[Q_2]_{\hat{\mathcal{N}}|_{Q_2}}$ is a colour propagating tree. The condition that $N_i\cap Q_1\not= \emptyset$ is trivially satisfied in case $N=N_i$ and follows as before in case $N\not= N_i$.
    \end{claimproof}
    \begin{claim}\label{claim:inductionStepCase2}
        Let $T$ be any spanning tree of $G_{\mathcal{M}}$.
        If there exists an edge $e\in E(G_{\mathcal{M}})\setminus E(T)$ then there is a partition $\hat{\mathcal{N}}=(\hat{\mathcal{N}}_1,\dots, \hat{\mathcal{N}}_{\hat{\ell}})$ of $V(G)$ with $\hat{\ell}\leq k+1$ and a coarsening $\mathcal{Q}=(Q_1,Q_2)$ of $\hat{\mathcal{N}}$ with the following properties. $|Q_i|$ is even and $\hat{\mathcal{N}}|_{Q_i}$ is a module partition of $G[Q_i]$ for $i\in [2]$. Furthermore, $G[Q_1]$ is connected and $G[Q_2]_{\hat{\mathcal{N}}|_{Q_2}}$ is either a star or a colour propagating tree.
    \end{claim}
    \begin{claimproof}
        For any edge $e=MM'\in E(G_{\mathcal{M}})\setminus E(T)$ we let $e_1, e_2\in E(T)$ such that $e_1$ is incident to $M$, $e_2$ is incident to $M'$, $(T\setminus\{e_1,e_2\})\cup \{e\}$ has exactly two components and $M$ is in the same component as  $M'$ in $(T\setminus\{e_1,e_2\})\cup \{e\}$. For $e=MM'\in E(G_{\mathcal{M}})\setminus E(T)$ let  $C_e$ be the subgraph of $G_\mathcal{M}$ induced by the vertices of the component of $(T\setminus\{e_1,e_2\})\cup \{e\}$ which contains $M$ and $M'$. We now define $e$ to be an edge minimizing $|V(C_e)|$. This means that $C_e$ must be a tree.
        First consider the case that $C_e$ is of even order and $C_e$ is a colour propagating tree. First consider that $G_\mathcal{M}\setminus C_e$ has at least  two vertices. In this case we can set $\hat{\mathcal{N}}:=\mathcal{M}$,  $Q_2:=\bigcup_{M\in V(C_e)}M$ and $Q_1:=V(G)\setminus Q_2$ satisfying all requirements. 
        
        Hence assume that $G_\mathcal{M}\setminus C_e$ consists of one vertex $N$. Hence in particular $e_1,e_2$ must be incident to $N$. Since $C_e$ is of even order $|N|$ must be even. Partition $N$ into two parts $N', N''$ of odd size and obtain $\hat{\mathcal{N}}$ from $\mathcal{M}$ by removing $N$ and adding $N'$ and $N''$. Furthermore, since $C_e$ is a colour propagating tree we get that $|\bigcup_{X\in V(X_e)} X|$ is odd where $X_e$ is one of the two components of $C_e$ after removing $e$. Now observe that since $N'$ is odd and adjacent to precisely one module of $X_e$ the graph $G_\mathcal{M}[V(X_e)\cup N']$ is a colour propagating tree. Hence we can set $Q_2:=N'\cup \bigcup_{X\in V(X_e)}X$ and $Q_1:=V(G)\setminus Q_2$ which satisfies all requirements.\\ 

        On the other hand, consider the case that $C_e$ is not a colour propagating tree. First assume that $C_e$ is of even order. Since $C_e$ is not a colour propagating tree 
        we can use \cref{claim:inductionStepCase1} on $C_e$ with module partition $\mathcal{M}|_{V_e}$ where  $V_e:=\bigcup_{N\in V(C_e)}N$.  We obtain a partition $\mathcal{N}'=(N_1',\dots, N_{\ell'}')$ of $V_e$ with $\ell'\leq |V(C_e)|+1$ and a coarsening $\mathcal{Q}'=(Q_1',Q_2')$ of $\mathcal{N}'$ such that  $M_i\cap Q_1'\not= \emptyset $ as in \cref{claim:inductionStepCase1} where $M_i$ is one of the modules incident to $e$. We obtain $\hat{\mathcal{N}}$ by removing all parts in $V(C_e)$ from $\mathcal{M}$ and adding the parts from $\mathcal{N}'$. We further set $Q_2:=Q_2'$ and $Q_1:=V(G)\setminus Q_2$. Note that since  $M_i\cap  Q_1'\not= \emptyset$, $G_{\mathcal{M}}\setminus C_e$, $G[Q_1']$ are connected and  either $e_1$ or $e_2$ is incident to both $M_i$ and some vertex in $G_{\mathcal{M}}\setminus C_e$ we get that $G[Q_1]$ is connected. All other properties follow from \cref{claim:inductionStepCase1}.
        
        On the other hand, if $C_e$ is of odd order then either $X_e$ or $Y_e$ must be of even order where $X_e, Y_e$ are the two connected components of $C_e$ after removing $e$. Without loss of generality  let $X_e$ be of even size.
        Note that removing $e$ and $e_1$  from $G_\mathcal{M}$ splits $G_\mathcal{M}$ into precisely two component of which one is $X_e$.  In the case that $X_e$ is a colour propagating tree or star we can set $\hat{\mathcal{N}}:=\mathcal{M}$,  $Q_2:=\bigcup_{M\in V(X_e)}M$ and $Q_1:=V(G)\setminus Q_2$. On the other hand, if $X_e$ is not a colour propagating tree we can
        use the same argument as above  only considering $X_e$ in place of $C_e$. 
    \end{claimproof}
    Note that since $G_{\mathcal{M}}$ is not a star or colour propagating tree the premise of either \cref{claim:inductionStepCase1} or \cref{claim:inductionStepCase2} must be satisfied. We obtain a partition $\hat{\mathcal{N}}=(\hat{N}_1,\dots,\hat{N}_{\hat{\ell}})$ of $V(G)$ with $\hat{\ell}\leq k+1$ and a coarsening $\mathcal{Q}=(Q_1,Q_2)$ of $\hat{\mathcal{N}}$ as in the two claims. Since $\hat{\mathcal{N}}|_{Q_2}$ must contain at least two modules we get that $\hat{\mathcal{N}}|_{Q_1}$ has strictly less modules than $\mathcal{M}$. Let $k'<k$ be the number of modules of $\hat{\mathcal{N}}|_{Q_1}$. Hence we can recursively obtain a partition $\mathcal{M}'$ of $G[Q_1]$ with at most $2k'$ parts and a coarsening $\mathcal{P}'$ of $\mathcal{M}'$ with the following properties.  $|P|$ is even, $\mathcal{M}'|_P$ is a module partition of $G[P]$ and $G[P]_{\mathcal{M}'|_{P}}$ is either a star or a colour propagating tree for every part $P$ of $\mathcal{P}$. We obtain the partition $\hat{M}$ of $V(G)$ by adding all parts of $\hat{N}|_{Q_2}$ to $\mathcal{M}'$ and the coarsening $\mathcal{P}$ of $\hat{\mathcal{M}}$ by adding $P_2$ to $\mathcal{P}'$. Note that the number of parts of $\hat{\mathcal{M}}$ is at most $2k'+(\hat{l}-k')\leq 2k$. Hence $\hat{\mathcal{M}}$ and $\mathcal{P}$ satisfy the conditions of the statement. 
\end{proof}
\begin{proof}[Proof of \cref{thm:oddChromaticModularWidth}]
    Without loss of generality assume that $G$ is connected. Furthermore, let $k:=\modularWidth(G)$ and $\mathcal{M}=(M_1,\dots, M_k)$ be a module partition of $G$. Let $\hat{\mathcal{M}}$ be a partition of $V(G)$ with at most $2k$ parts and $\mathcal{P}$ a refinement of $\hat{\mathcal{M}}$ as in \cref{lem:decompositionOfModuleGraph}. First observe that $\hat{\mathcal{M}}|_P$ must contain at least two parts for every part $P$ of $\mathcal{P}$ as $\hat{\mathcal{M}}|_P$ is a module partition of $G[P]$. Since $\hat{\mathcal{M}}$ has at most $2k$ parts and $\mathcal{P}$ is a refinement of $\hat{\mathcal{P}}$ this implies that $\mathcal{P}$ has at most $k$ parts. Since $G[P]_{\hat{\mathcal{M}}|_P}$ is either a star or a colour propagating tree we get that $\chiodd(G[P])\leq 3$ for every part $P$ of $\mathcal{P}$ by \cref{lem:colouringStars} and \cref{lem:colouringTrees}.  Using a partition $(W_1^P,W_2^P,W_3^P)$ of $G[P]$ such that $G[W_i^P]$ is odd for every $i\in [3]$ for every part $P$ we obtain a global partition  of $G$ into at most $3k$ parts such that each part induces an odd subgraph.
\end{proof}
Since deciding whether a graph is $k$-odd colourable can be solved in time $2^{\mathcal{O}(k\operatorname{rw}(G))}$ \cite[Theorem 6]{BelmonteS21} and $rw(G)\leq \operatorname{cw}(G)\leq \modularWidth(G)$, where $\operatorname{cw}(G)$ denotes the clique-width of $G$ and  $\operatorname{rw}(G)$ rank-width, we obtain the following as a corollary.
\begin{corollary}
    Given a graph $G$ and a module partition of $G$ of width $m$ the problem of deciding whether $G$ can be odd coloured with at most $q$ colours can be solved in time $2^{\mathcal{O}(m^2)}$.
\end{corollary}

\section{Interval graphs}
\label{sec:interval}

In this section we study the odd chromatic number of interval graphs and provide an upper bound in the general case as well as a tight upper bound in the case of proper interval graphs.
We use the following lemma in both proofs.

\begin{lemma}\label{claim:pathCoversV(G)}
    Let $G$ be a connected interval graph and $P=(p_1,\dots,p_k)$ a maximal induced path in $G$ with the following property.
    \begin{enumerate}[left=3pt , label=$(\ast)$]
        \item\label{goodSelectionOfPath} $\ell_{p_1}=\min\{\ell_v:v\in V(G)\}$ and for every $i\in[k-1]$ we have that $r_{p_{i+1}}\geq r_v$ for every $v\in N_G(p_{i})$.
    \end{enumerate}
    Then every $v\in V(G)$ is adjacent to at least one vertex on $P$.
\end{lemma}
   \begin{proof}
        Towards a contradiction, assume that there is $v\in V(G)$ such that $v$ is not adjacent to any vertex of $P$. Note that $v\notin \{p_1,\dots,p_k\}$. 
        Furthermore, by the assumption that $v$ is not adjacent to any vertex of $P$ either $\ell_{p_i}\leq r_{p_i}<\ell_v$ or $\ell_v<\ell_{p_i}\leq r_{p_i}$ for every $i\in [k]$.   Pick $i\in [k]$ to be the maximum index such that $r_{p_i}<\ell_v$. Observe that $i$ is well defined as by property \ref{goodSelectionOfPath} 
         $\ell_{p_1}=\min \{\ell_v: v\in V(G)\}\leq r_{p_1}< \ell_v$. First consider the case that $i<k$. But then  $r_{p_i}<\ell_v\leq r_v< \ell_{p_{i+1}}$ which contradicts that $p_i$ and $p_{i+1}$ are adjacent. Hence $i=k$. Since $G$ is connected there must be a path $Q=(q_1,\dots,q_\ell)$ from $p_k$ to $v$. Let $j\in [\ell]$ be the last index such that $\ell_{q_j}\leq r_{p_k}$. 
         Since $q_1=p_k$ we know that $q_{j}$ exists and is adjacent to some vertex in $P$. Indeed $j<\ell$  as   $r_{p_k}<\ell_v$ and $q_\ell=v$.
         Therefore $q_{j+1}$ exists and further $q_{j}q_{j+1}\in E(G)$ and $\ell_{q_{j+1}}> r_{p_k}$ (by choice of $j$). We conclude that  $r_{p_k}<r_{q_{j}}$. If $q_{j}$ is adjacent to $p_{k-1}$ this contradicts the property \ref{goodSelectionOfPath}. On the other hand, if $q_{j}$ is not adjacent to $p_{k-1}$ then the set $\{v\in V(G):p_{k-1}v\notin E(G),p_kv\in E(G) \}$ is not empty which contradicts the maximality of $P$. Hence $v$ has to be adjacent to at least one vertex of $P$.
    \end{proof}
To prove that the odd chromatic number of proper interval graphs is bounded by three we essentially partition the graph into maximal even sized cliques greedily in a left to right fashion.
\begin{proposition}\label{prop:chiOddProperInterval}
    For every proper interval graph $G$ with all components of even order $\chiodd(G)$ is at most three.
\end{proposition}
\begin{proof}
    We assume that $G$ is connected.
    Fix  an interval representation of $G$ and denote the interval representing vertex $v\in V(G)$ by $I_v=[\ell_v,r_v]$ where $\ell_v, r_v\in \mathbb{R}$.  Let $P=(p_1,\dots, p_k)$ be a maximal induced path in $G$ as in \cref{claim:pathCoversV(G)}.
    For every vertex $v\in V(G)\setminus \{p_1,\dots, p_k\}$ let $i_v \in [k]$ be the index such that $p_{i_v}$ is the first neighbour of $v$ on $P$. Note that this is well defined by \cref{claim:pathCoversV(G)}.  For  $i\in [k]$ we let $Y_i$ be the set with the following properties.
    \begin{enumerate}[left=5pt , label=$(\Pi\arabic*)_i$]
        \item\label{item:pi1} $\{v\in V(G): i_v=i\}\subseteq Y_i \subseteq \{v\in V(G): i_v=i\} \cup \{p_i,p_{i+1}\}$ .
        \item\label{item:pi2} $p_i\in Y_i$ if and only if $\big|\{p_1,\dots, p_{i-1}\}\cup \bigcup_{j\in [i-1]}\{v\in V(G):i_v=j\}\big|$ is even.
        \item\label{item:pi3} $p_{i+1}\in Y_{i}$ if and only if $\big|\{p_1,\dots, p_{i}\}\cup \bigcup_{j\in [i]}\{v\in V(G):i_v=j\}\big|$ is odd.
    \end{enumerate}
    First observe that $(Y_1,\dots,Y_k)$ is a partition of $V(G)$ as \ref{item:pi2} and \ref{item:pi3} imply that every $p_i$ is in exactly one set $Y_i$.
    Furthermore, $|Y_i|$ is even for every $i\in [k]$ since  \ref{item:pi1} and \ref{item:pi3} imply that $\big|Y_i\cup \{p_1,\dots, p_{i}\}\cup \bigcup_{j\in [i-1]}\{v\in V(G):i_v=j\}\big|$ is even and  \ref{item:pi2} implies that  $\big|(\{p_1,\dots, p_{i}\}\cup \bigcup_{j\in [i-1]}\{v\in V(G):i_v=j\})\setminus Y_i\big|$ is even.  Since $v\in V(G)\setminus \{p_1,\dots,p_k\}$ is not adjacent to $p_{i_v-1}$ we get that $\ell_v\in I_{p_{i_v}}$. Since $G$ is a proper interval graph this implies that $r_{p_{i_v}}\leq r_v$ and hence $v$ is adjacent to $p_{i_v+1}$. Hence \ref{item:pi1} implies that $G[Y_i]$ must be a clique since $Y_i\cap \{p_1,\dots,p_k\}\subseteq \{p_i,p_{i+1}\}$ for every $i\in [k]$. 
    Furthermore, $N_G(Y_i)$ and $ Y_{i+3}$ are disjoint since $r_v\leq r_{p_{i+1}}$ for every $v\in Y_i$ by property \ref{goodSelectionOfPath} and $r_{p_{i+1}}<\ell_{p_{i+3}}\leq r_w$ for every $w\in Y_{i+3}$ since $P$ is induced. Hence we can define an odd-colouring $(V_1,V_2,V_3)$ of $G$ in the following way. We let $V_j:=\bigcup_{i\equiv j\mod 3}Y_i$ for $j\in[3]$. Note that since $N_G(Y_i)\cap Y_{i+3}$ we get that $\deg_{G[Y_i]}(v)=\deg_{G[V_j]}(v)$ for $i\equiv j \mod 3$ which is odd (as $Y_i$ is a clique of even size). Hence $G[V_j]$ is odd for every $j\in [3]$.  
\end{proof}
\begin{remark}
    The upper bound presented in \cref{prop:chiOddProperInterval} is tight. Consider the graph $G$ consisting of a $K_4$ with two added pendant vertices $u,w$ adjacent to different vertices of $K_4$. Clearly, $G$ is a proper interval graph and further $\chiodd(G)=3$.
\end{remark}
We use a similar setup (\ie, a path $P$ covering all vertices of the graph $G$) as in the proof of \ref{prop:chiOddProperInterval} to show our general upper bound for interval graphs. The major difference is that we are not guaranteed that sets of the form $\{p_i\}\cup \{v\in V(G):i_v=i\}$ are cliques. To nevertheless find an odd colouring with few colours of such sets we use an odd/even colouring as in \ref{thm:Gallai} of $\{v\in V(G):i_v=i\}$ and the universality of $p_i$. Hence this introduces a factor of two on the number of colours. Furthermore, this approach prohibits us from moving the $p_i$ around as in the proof of \cref{prop:chiOddProperInterval}. As a consequence we get that the intervals of vertices contained in a set $Y_i$ span a larger area of the real line than they do in the proof of \cref{prop:chiOddProperInterval}. This makes the analysis more technical. 
\begin{theorem}\label{thm:oddChromaticInterval}
    For every interval graph $G$ with all components of even order $\chiodd(G)$ is at most six.
\end{theorem}
\begin{proof}
    We assume that $G$ is connected. 
    First we fix  an interval representation of $G$. We denote the interval representing vertex $v\in V(G)$ by $I_v=[\ell_v,r_v]$ where  $\ell_v, r_v\in \mathbb{R}$. 
     Let $P=(p_1,\dots, p_k)$ be a maximal induced path in $G$ as in \cref{claim:pathCoversV(G)}.
    Let $Y$ be $V(G)\setminus \{p_1,\dots,p_k\}$. For every $v\in Y$ we define $i_v\in [k]$ to be the minimum index such that $v$ is adjacent to $p_{i_v}$. Note that this is well defined by \cref{claim:pathCoversV(G)}. 
    
    We now recursively define a partition $(Y_1,\dots,Y_{k})$ of $Y$ such that for every $i\in [k]$ the following properties hold.
        \begin{enumerate}[left=4pt , label=$(P\arabic*)_i$]
            \item\label{P1} Every vertex in $Y_i$ is adjacent to $p_i$.
            \item\label{P2} If $|\{p_1,\dots, p_i\}\cup \{v\in Y: i_v\leq i\}|$ is even then  $\bigcup_{j=1}^{i}Y_{j}= \{v\in Y:i_v\leq i\}$.
            \item\label{P3} If $|\{p_1,\dots, p_i\}\cup \{v\in Y: i_v\leq i\}|$ is odd then either $| \{v\in Y:i_v\leq i\}\setminus \bigcup_{j=1}^{i}Y_{j}|= 1$ or $\bigcup_{j=1}^{i}Y_{j}= \{v\in Y:i_v\leq i\}$ and $N_G(p_{i+1})\cap Y_i=\emptyset$.
            \item\label{P4}  If $i_w\leq i$ for   $w\notin \bigcup_{j=1}^{i}Y_j$ then $w\in N_G(p_{i+1})$ and $i_w=\max\{i_v: v\in Y\cap N_G(p_{i+1})\}$.
         \end{enumerate}
     Fix $i\in [k]$ and assume that we have defined $Y_1,\dots,Y_{i-1}$  satisfying \hyperref[P1]{$(P1)_j$}, \hyperref[P2]{$(P2)_j$}, \hyperref[P3]{$(P3)_j$}, \hyperref[P4]{$(P4)_j$} for every $j\in [i-1]$. In the following we show how to construct $Y_i$.
     Define  $Y_i':=\big\{v\in Y\setminus \bigcup_{j=1}^{i-1}Y_j:i_v\leq i\big\}$. Note that $Y_i'\cup \bigcup_{j=1}^{i-1}Y_j=\{v\in Y: i_v\leq i\}$.
    In the case that either $|\{p_1,\dots, p_i\}\cup \{v\in Y: i_v\leq i\}|$ is even or  $Y_i'\cap N_G(p_{i+1})=\emptyset$ we set $Y_{i}:=Y_i'$. Otherwise, pick  $w\in Y_i'\cap N_G(p_{i+1})$  such that $i_w=\max\{i_v: v\in Y_i'\cap N_G(p_{i+1})\}$ and define $Y_i:=Y_i'\setminus \{w\}$. Note that $w$ is well defined since we are considering the case that $Y_i'\cap N_G(p_{i+1})\not=\emptyset$. Observe that properties \ref{P2} and \ref{P3}  are true by construction of $Y_i$. To argue that property \ref{P1} is true we observe that by \hyperref[P4]{$(P4)_{i-1}$} every vertex $w\in Y_i'$ with $i_w< i$  has to be adjacent to $p_i$. Since in addition every vertex $v$ with $i_v=i$ is adjacent to $p_i$ by choice of $i_v$, property \ref{P1} holds. To argue that property \ref{P4} holds we observe that every vertex $v\in Y$ with $i_v=i$ is contained in $Y_i'$. Hence if $\max\{i_v: v\in Y\cap N_G(p_{i+1})\}=i$ then we would choose $w$ with $i_w=i$. In the case that $\max\{i_v: v\in Y\cap N_G(p_{i+1})\}<i$ then \ref{P4} follows directly from \hyperref[P4]{$(P4)_{i-1}$}. This concludes the construction of the sets $Y_1,\dots, Y_k$. The following two claims  allows us to reuse the colours used to colour $Y_i$ for sets $Y_{i+c},Y_{i+2c},\dots$ for some small constant $c$.

    \begin{claim}\label{claim:atMostFourNeighboursOnP}
        For every vertex $v\in Y$ it holds that
       $  I_v\cap I_{p_i}=\emptyset$ for every $i\notin \{i_v,i_v+1, i_v+2\}$. 
        In particular, $N_G(v)\cap \{p_1,\dots,p_k\}$ is contained in $\{p_{i_v},p_{i_v+1},p_{i_v+2}\}$ for every vertex $v\in Y$.
    \end{claim}
    \begin{claimproof}
        First observe that $I_v\cap I_{p_i}=\emptyset$ for every $i<i_v$ by definition of $i_v$. 
        Since $P$ is an induced path $r_{p_{i_v+1}}< \ell_{p_{i_v+3}}$. On the other hand, $r_v\leq r_{p_{i_v+1}}$ by property \ref{goodSelectionOfPath}. Hence $r_v<\ell_{p_{i_v+3}}\leq \ell_{p_i}$ for every $i\geq i_v+3$. Hence $I_v\cap I_{p_i}=\emptyset$ for every $i\geq i_v+3$  concluding the proof of the statement.
    \end{claimproof}
    As a consequence of \cref{claim:atMostFourNeighboursOnP} we get the following claim.
    \begin{claim}\label{claim:reusageOfColours}
        If $Y_i\subseteq \{v\in Y:i_v\geq i'\}$ then 
        \begin{itemize}
            \item $N_G(\{p_i\} )$ is disjoint from $\{p_{j}\} \cup Y_{j}$ for any $j\leq i-3$ and 
            \item $N_G( Y_i )$ is disjoint from $\{p_{j}\} \cup Y_{j}$ for any $j\leq i'-2$.
        \end{itemize}
    \end{claim}
    \begin{claimproof}
        From \cref{claim:atMostFourNeighboursOnP} we get that  no $v\in Y$ with $i_v\leq i-3$ can be adjacent to  $\{p_i\}$. Furthermore, $P$ is an induced path so $p_j$ is non-adjacent to $p_{i}$ for every $j\leq i-2$. Therefore, $N_G(\{p_i\})$ is disjoint from $\{p_j\}\cup Y_j$ for every $j\leq  i-3$.
        
        To prove the second property, observe that the property \ref{goodSelectionOfPath} and $P$ being an induced path imply that for every $v\in \{p_1,\dots, p_{i'-2}\}\bigcup_{j=1}^{i'-2}$ we have that $r_v\leq r_{p_{i'-1}}$. Since every $w\in Y_i$ satisfies that $i_w\geq i'$ we get that every $v\in \{p_1,\dots, p_{i'-2}\}\bigcup_{j=1}^{i'-2}Y_j$ cannot be adjacent to any vertex in $Y_i$ (note that this is not true for $p_i$ in case $i=i'$ as $r_{p_i}\leq r_{p_{i-1}}$). 
        Since by construction of $Y_1,\dots, Y_k$ for every $j\leq i'-2$ we have $Y_j\subseteq \{v\in Y:i_v\leq i'-2\}$ we get that $N_G(Y_i)$ is disjoint from $Y_j$ for any $j\leq i'-2$. 
    \end{claimproof}
    Using the sets $Y_1,\dots, Y_k$, \cref{claim:atMostFourNeighboursOnP} and \cref{claim:reusageOfColours} we can now find an odd colouring of $G$.  To colour $G$  we use a recursive argument. In the $i$-th step we find a partition of the set $\{p_1,\dots, p_{i}\}\cup \bigcup_{j=1}^{i}Y_j $ into six (possibly empty) parts $V_1^i,\dots, V_6^i$ with the following properties.
    \begin{enumerate}[left=4pt , label=$(C\arabic*)_i$]
        \item\label{C1} If $|\{p_1,\dots, p_{i}\}\cup \bigcup_{j=1}^{i}Y_j| $ is even then $(V_1^i,\dots,V_{6}^i)$ is an odd colouring of $G[\{p_1,\dots, p_{i}\}\cup \bigcup_{j=1}^{i}Y_j]$.
        \item\label{C2} If $|\{p_1,\dots, p_{i}\}\cup \bigcup_{j=1}^{i}Y_j| $ is odd then there is $j_i\in [6]$ such that $G[V_j^i]$ is odd for every $j\not= j_i$, $p_i\in V_{j_i}^i$ and in $G[V_{j_i}^i]$ every vertex apart from $p_i$ has odd degree.
        \item\label{C3} $\{p_i\} \cup Y_i$ is contained in the union of at most two parts of the partition $(V_1^i,\dots, V_{6}^i)$.
        \item\label{C4} For every $j\in [6]$, any pair of vertices $v,w\in V_j^i$ can be separated in $G[V_j^i]$ by removing an edge of the path $P$ if there are two indices $i'\not=i''$ such that $v\in \{p_{i'}\}\cup Y_{i'}$ and $w\in \{p_{i''}\}\cup Y_{i''}$. 
    \end{enumerate}
    Let us fix $i\in [k]$ and assume we have partitioned $\{p_1,\dots, p_{i-1}\}\cup \bigcup_{j=1}^{i-1}Y_j $ into six parts $V_1^{i-1},\dots, V_{6}^{i-1}$ with properties \hyperref[C1]{$(C1)_{j}$}, \hyperref[C2]{$(C2)_{j}$}, \hyperref[C3]{$(C3)_{j}$} and \hyperref[C4]{$(C4)_{j}$} for every $j\leq i-1$. Our goal is to find a partition $(W_1,W_2)$ of $\{p_i\} \cup Y_i$ and two indices $j_1\not= j_2\in [6]$ such that  the partition obtained from $(V_1^{i-1},\dots,V_{6}^{i-1})$ by adding $W_1$ to $V_{j_1}^{i-1}$ and $W_2$ to $V_{j_2}^{i-1}$ is a partition of $\{p_1,\dots, p_{i}\}\cup \bigcup_{j=1}^{i}Y_j$ with properties \ref{C1}, \ref{C2}, \ref{C3} and \ref{C4}. 

    To define the partition $(W_1,W_2)$ we use  a partition $(W_1',W_2')$ of $G[Y_i]$ such that $G[W_1']$ is odd and $G[W_2']$ is even which exists due to \cref{thm:Gallai}. Note that in the case that $Y_i$ is empty we simply obtain the partition $(\emptyset,\emptyset)$ which is sufficient for our purpose. We define $W_1:= W_1'$ and $W_2:= W_2'\cup \{p_i\}$. Observe that $G[W_1']$ being odd implies that $|W_1'|$ is even by the handshake lemma. Hence $|W_2'|$ is odd if and only if $|Y_i|$ is odd. Since every vertex in $W_2'$ is adjacent to $p_1$ by \ref{P1},  we obtain that every vertex in $W_2'$ has odd degree in $G[W_2]$ and $p_1$ has odd degree in $G[W_2]$ if and only if $|Y_i|$ is odd. Note that we can get a colouring with 12 colours at this point without much further analysis. 
    Obtaining a colouring with six colours requires careful analysis.
    
    The following claim will provide us with possible choices for  indices $j_1$ and $j_2$. Note that  the indices from the claim will not in every case be a suitable choice.
    
    \begin{claim}\label{claim:availableColours}
        There is a partition $(\hat{V}_1^{i-1},\dots, \hat{V}_{6}^{i-1})$ of $\{p_1,\dots, p_{i-1}\}\cup \bigcup_{j=1}^{i-1}Y_j $ with properties \hyperref[C1]{$(C1)_{j}$}, \hyperref[C2]{$(C2)_{j}$}, \hyperref[C3]{$(C3)_{j}$}, \hyperref[C4]{$(C4)_{j}$} for every $j\leq i-1$ and  indices $\hat{j}_1\not=\hat{j}_2\in [6]$ such that $N_G(W_1)$ is disjoint from $\hat{V}_{\hat{j}_1}^{i-1}$ and $N_G(W_2)$ is disjoint from $\hat{V}_{\hat{j}_2}^{i-1}$.
    \end{claim}
    \begin{claimproof}
        First consider the case that $Y_i\subseteq \{v\in Y: i_v\in \{i,i-1\}\}$.
        In this case we will set $\hat{V}_{j}^{i-1}:=V_j^{i-1}$ for every $j\in [6]$.  
        Since $Y_i\subseteq \{v\in Y: i_v\in \{i,i-1\}\}$   we obtain using \cref{claim:reusageOfColours} that both $N_G(W_1)$ and $N_G(W_2)$ are disjoint from  any set $\{p_j\}\cup Y_j$ with $j\leq i-3$. Since  \hyperref[C3]{$(C3)_j$}, $j<i$ implies that $\{p_{i-2},p_{i-1}\}\cup Y_{i-2}\cup Y_{i-1}$ is contained in at most four parts of the partition $(\hat{V}_1^{i-1},\dots, \hat{V}_{6}^{i-1})$ we can find $\hat{j}_1,\hat{j}_2\in [6]$, $\hat{j}_1\not= \hat{j}_2$ with the following properties.
        $\hat{V}_{\hat{j}_1}^{i-1}$ and $\hat{V}_{\hat{j}_2}^{i-1}$ do not contain any element from $\{p_{i-2}, p_{i-1}\}\cup Y_{i-2}\cup Y_{i-1}$. This choice guarantees that $N_G(W_1)$ is disjoint from $\hat{V}_{\hat{j}_1}^{i-1}$ and $N_G(W_2)$ is disjoint from $\hat{V}_{\hat{j}_2}^{i-1}$.\\ 

        \begin{figure}
            \centering
            \includegraphics[scale=1.3]{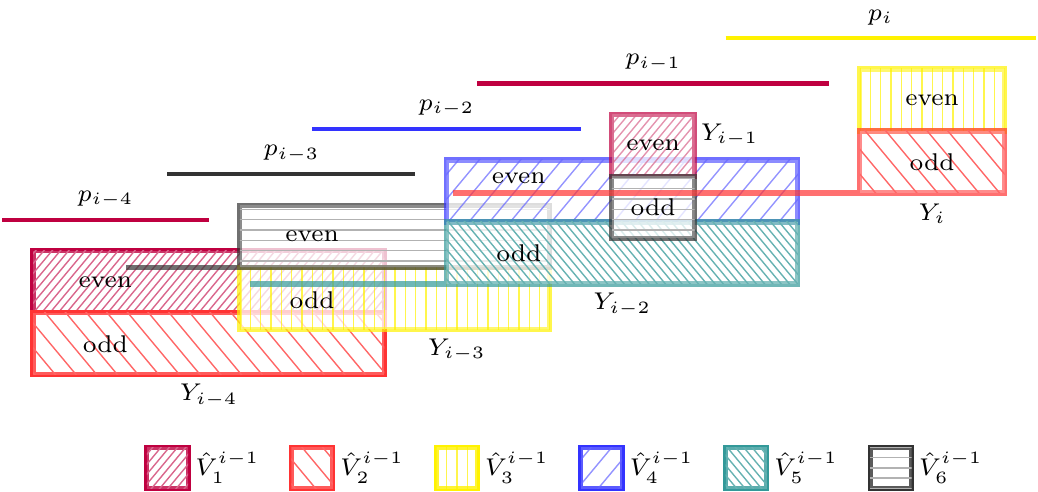}
            \vspace*{2pt}
            \caption{Schematic representation of the partition $(\hat{V}_1^{i-1},\dots,\hat{V}_6^{i-1})$  in the case that $i_w=i-2$ in the proof of \cref{claim:availableColours}. Note that in the figure $Y_i$ is coloured in the two colours of the indices $\hat{j}_1=3$ and $\hat{j}_2=2$ we obtain in this case.}
            \label{fig:colouringInterval}
        \end{figure}

         Now consider the case that $Y_i$ contains a vertex $w$ with $i_w<i-1$. Observe that in this case we get   that $i_w=i-2$ as a consequence of \cref{claim:atMostFourNeighboursOnP}. \Cref{fig:colouringInterval} illustrates the layout of intervals and the available colours we obtain in this case. 
         
         We know that at least one set out of $W_1,W_2$ is fully contained in $\{v\in Y: i_v=i\}\cup\{p_i\}$ (\ie, the one not containing $w$) by \ref{P3}. Without loss of generality assume that this is true for $W_1$ (\ie, we do not use in the following argument that $W_2$ contains $p_i$).
         By \cref{claim:reusageOfColours}  we infer that $N_G(W_2)\subseteq N_G(\{p_i\}\cup Y_i)$ is disjoint from  any set $\{p_j\}\cup Y_j$ with $j\leq i-4$. Furthermore,  \hyperref[C3]{$(C3)_j$}, $j<i$ implies that $\{p_{i-3},p_{i-2}\}\cup Y_{i-3}\cup Y_{i-2}$ is contained in at most four parts of the partition $(V_1^{i-1},\dots, V_{6}^{i-1})$. Pick $\hat{j}_2$ in such a way that $V_{\hat{j}_2}^{i-1}$ is disjoint from $\{p_{i-3},p_{i-2}\}\cup Y_{i-3}\cup Y_{i-2}$. 
         In the following we will argue that we can, after potentially modifying the partition $(V_1^{i-1},\dots, V_{6}^{i-1})$, assume that $\{p_{i-1}\}\cup Y_{i-1}$ is also disjoint from $V_{\hat{j}_2}^{i-1}$. 

         First observe that $w\in Y_i$ with $i_w=i-2$ implies that $Y_{i-1}\subseteq \{v\in Y: i_v=i-1\}$ by \hyperlink{P2}{$(P2)_{i-2}$} and \hyperlink{P3}{$(P3)_{i-2}$}. By \cref{claim:reusageOfColours} we get that $N_G(\{p_{i-1}\}\cup Y_{i-1})$ is disjoint from  $\{p_{j}\}\cup Y_j$ for any $j\leq i-3$.  Since $\{ p_{i-2}\}\cup  Y_{i-2}$ is contained in at most two parts of the partition $(V_1^{i-1},\dots,V_6^{i-1})$ by \hyperref[C3]{$(C3)_j$}, $j<i$ we can pick $j',j''\in [6]$, $j'\not=j''$ distinct from $\hat{j}_2$ such that $V_{j'}^{i-1}$, $V_{j''}^{i-1}$ are disjoint from $\{p_{i-2}\}\cup  Y_{i-2}$.
        By \hyperref[C3]{$(C3)_{i-1}$} we know that $\{p_{i-1}\}\cup  Y_{i-1}$ is contained in at most two parts of the partition $(V_1^{i-1},\dots, V_{6}^{i-1})$. Assume $V_{m'}^{i-1}$, $V_{m''}^{i-1}$ are those two parts and define $M':=V_{m'}^{i-1}\cap (\{p_{i-1}\}\cup  Y_{i-1})$ and $M'':=V_{m''}^{i-1}\cap (\{p_{i-1}\}\cup  Y_{i-1})$. In case $\{p_{i-1}\}\cup Y_{i-1}$ is contained in one part we let $m'=m''$.  If neither $m'=\hat{j}_2$ nor $m''=\hat{j}_2$ then $V_{\hat{j}_2}^{i-1}$ is disjoint from $\{p_{i-3},P_{i-2},p_{i-1}\}\cup \bigcup_{j=i-3}^{i-1}Y_j$ and hence with setting $\hat{V}_j^{i-1}:=V_j^{i-1}$ for every $j\in [6]$ we get that $N_G(W_2)$ is disjoint from $\hat{V}_{\hat{j}_2}^{i-1}$. Hence assume that this is not the case. Without loss of generality assume $m'=\hat{j}_2$. Since $j'\not=j''$ we get that $m''$ is not equal to either $j'$ or $j''$. Assume $j'\not=m''$. We define $\hat{V}_{j'}^{i-1}:=V_{j'}^{i-1}\cup M'$, $\hat{V}_{\hat{j}_2}^{i-1}:=V_{\hat{j}_2}^{i-1}\setminus M'$ and $\hat{V}_j^{i-1}:=V_j^{i-1}$ for $j\notin \{j',\hat{j}_2\}$. To see that the partition $(\hat{V}_1^{i-1},\dots,\hat{V}_6^{i-1})$ of $\{p_1,\dots,p_{i-1}\}\cup \bigcup_{j=1}^{i-1}Y_j$ satisfies  \hyperref[C1]{$(C1)_{i-1}$}, \hyperref[C2]{$(C2)_{i-1}$}, \hyperref[C3]{$(C3)_{i-1}$}, \hyperref[C4]{$(C4)_{i-1}$} we make the following two observations. Since $V_{\hat{j}_2}^{i-1}$ is disjoint from $\{p_{i-3},p_{i-2}\}\cup Y_{i-3}\cup Y_{i-2}$ and $N_G(M')\subseteq N_G(\{p_{i-1}\}\cup  Y_{i-1})$ is disjoint from $\{p_{j}\}\cup Y_j$ for any $j\leq i-3$ we get that $M'$ is a component of $G[V_{\hat{j}_2}^{i-1}]$. Hence $\deg_{G[V_{\hat{j}_2}^{i-1}]}(v)=\deg_{G[\hat{V}_{\hat{j}_2}^{i-1}]}(v)$ for any vertex $v\in V_{\hat{j}_2}^{i-1}\setminus M'$. Additionally,  since $V_{j'}^{i-1}$ is disjoint from $\{p_{i-2}\}\cup Y_{i-2}$ and $N_G(M')\subseteq N_G(\{p_{i-1}\}\cup Y_{i-1})$ is disjoint from $\{p_{j}\}\cup Y_j$ for any $j\leq i-3$ we get that $M'$ is a connected component of $G[\hat{V}_{j'}^{i-1}]$. Hence $\deg_{G[V_{j'}^{i-1}]}(v)=\deg_{G[\hat{V}_{j'}^{i-1}]}(v)$ for any vertex $v\in V_{j'}^{i-1}$ and $\deg_{G[V_{\hat{j}_2}^{i-1}]}(v)=\deg_{G[\hat{V}_{j'}^{i-1}]}(v)$ for any vertex $v\in M'$. This argument shows that \hyperref[C1]{$(C1)_{i-1}$} and \hyperref[C2]{$(C2)_{i-1}$} are satisfied.   \hyperref[C4]{$(C4)_{i-1}$} follows from the observation that $M'$ is a connected component of $G[\hat{V}_{j'}^{i-1}]$ and $(V_1^{i-1},\dots,V_6^{i-1})$ satisfying \hyperref[C4]{$(C4)_{i-1}$}. Furthermore, \hyperref[C3]{$(C3)_{i-1}$} is trivially satisfied.

        To choose $\hat{j}_1$ we first observe that $W_1\subseteq \{p_i\}\cup \{v\in Y: i_v=i\}$ implies  that $N_G(W_1)$ is disjoint from $\{p_j\}\cup Y_j$ for every $j\leq i-3$ using \cref{claim:reusageOfColours}. Since $\{p_{i-2},p_{i-1}\}\cup Y_{i-2}\cup Y_{i-1}$ is contained in at most four parts of the partition $(\hat{V}_1^{i-1},\dots, \hat{V}_{6}^{i-1})$ by \hyperref[C3]{$(C3)_{j}$}, $j<i$ we can choose $\hat{j}_1\in [6]$ such that $\hat{j}_1\not=\hat{j}_2$ and  $N_G(W_1)$ is disjoint from $\hat{V}_{\hat{j}_1}^{i-1}$ as required.
    \end{claimproof}
    For the remainder of the argument we pick $\hat{j}_1,\hat{j}_2\in [6]$ and $(\hat{V}_1^{i-1},\dots, \hat{V}_{6}^{i-1})$ as in the statement of \cref{claim:availableColours}. We define partition $(V_1^i,\dots,V_6^i)$ in the following considering several different cases.
    
    First consider the case that $|\{p_1,\dots, p_{i-1}\}\cup \bigcup_{j=1}^{i-1}Y_j| $ is even.  We set $j_1:=\hat{j}_1$ and $j_2:=\hat{j}_2$ and define $V_{j_1}^i:=\hat{V}_{j_1}^{i-1}\cup W_1$, $V_{j_2}^i:=\hat{V}_{j_2}^{i-1}\cup W_2$ and $V_{j}^i:=\hat{V}_{j}^{i-1}$ for every $j\notin\{j_1,j_2\}$. Since $|\{p_1,\dots, p_{i-1}\}\cup \bigcup_{j=1}^{i-1}Y_j| $ is even we get that $G[\hat{V}_{j}^{i-1}]$ is odd for every $j\in [6]$ by \hyperref[C1]{$(C1)_{i-1}$}. Since additionally $G[W_1]$ is odd and 
    $N_G(W_1)$ is disjoint from $\hat{V}_{j_1}^{i-1}$ we get that $G[V_{j_1}^i]$ is odd. Furthermore, recall that every vertex $v\in W_2\setminus \{p_i\}$ has odd degree in $G[W_2]$ and $p_i$ has odd degree in $G[W_2]$ if and only if $|Y_i|$ is even. Since $|\{p_1,\dots, p_{i-1}\}\cup \bigcup_{j=1}^{i-1}Y_j| $ is even, $|\{p_1,\dots, p_{i}\}\cup \bigcup_{j=1}^{i}Y_j| $ is even if and only if $|Y_i|$ is even (and $p_i$ has odd degree). Hence $(V_1^i,\dots, V_6^i)$ satisfies \ref{C1} and \ref{C2}. Additionally, \ref{C3} is true because $\{p_i\}\cup Y_i$ are contained in $V_{j_1}^i\cup V_{j_2}^i$. Lastly, \ref{C4} follows from \hyperref[C4]{$(C4)_{i-1}$} and the fact that $N_G(W_1)$ is disjoint from $\hat{V}_{j_1}^{i-1}$ and $N_G(W_2)$ is disjoint from $\hat{V}_{j_2}^{i-1}$.\\
    
    Now consider the case that $|\{p_1,\dots, p_{i-1}\}\cup \bigcup_{j=1}^{i-1}Y_j| $ is odd.  Assume $j_2\in[6]$ is the index such that $p_{i-1}\in \hat{V}_{j_2}^{i-1}$ and further set $j_1:=\hat{j_1}$. We define $V_{j_1}^i:=\hat{V}_{j_1}^{i-1}\cup W_1$, $V_{j_2}^i:=\hat{V}_{j_2}^{i-1}\cup W_2$ and $V_{j}^i:=\hat{V}_{j}^{i-1}$ for every $j\notin\{j_1,j_2\}$. By \hyperref[C1]{$(C1)_{i-1}$} and \hyperref[C2]{$(C2)_{i-1}$} we directly conclude that $G[V_j^i]$ is odd for $j\notin \{j_1,j_2\}$.

    Note that by \cref{claim:availableColours} we have that $N_G(W_1)$ is disjoint from $\hat{V}_{j_1}^{i-1}$. As $W_1$ is therefore a connected component in $G[V_{j_1}^i]$ and both $G[W_1]$ and $G[\hat{V}_{j_1}^{i-1}]$ are odd we get that $G[V_{j_1}^i]$ is odd. Furthermore, property \hyperref[C2]{$(C2)_{i-1}$} implies that $\deg_{G[\hat{V}_{j_2}^{i-1}]}(p_{i-1})$ is even and $\deg_{G[\hat{V}_{j_2}^{i-1}]}(v)$ is odd for every $v\in \hat{V}_{j_2}^{i-1}\setminus \{p_{i-1}\}$. Additionally, $\deg_{G[W_2]}(v)$ is odd for every $v\in Y_i\setminus \{p_i\}$.  
    To determine the degree of $p_{i-1}$ in $G[V_{j_2}^i]$ observe that \hyperref[P3]{$(P3)_{i-1}$} ensures that $Y_i\subseteq \{v\in Y: i_v=i\}$ 
    and hence $p_{i-1}$ is non-adjacent to any $v\in Y_i$.  Since $p_{i-1}$ is adjacent to $p_i$ we get $\deg_{G[V_{j_2}^{i}]}(p_{i-1})=\deg_{G[\hat{V}_{j_2}^{i-1}]}(p_{i-1})+1$ and hence the degree of $p_{i-1}$ is odd in $G[V_{j_2}^i]$. To determine the degree of $v\in V_{j_2}^i\setminus\{p_{i-1}\}$, we first argue that  \ref{C4} holds for the partition $(V_1^i,\dots,V_6^k)$. Let $G'$ be the graph obtained from $G[\hat{V}_{j_2}^{i-1}]$ by removing all edges of the path $P$.  Then \hyperref[C4]{$(C4)_{i-1}$} implies that every $v\in \hat{V}_{j_2}^{i-1}\cap (\{p_1,\dots, p_{i-2}\}\cup \bigcup_{j=1}^{i-2}Y_j)$ has to be in a different component of $G'$ then any $w\in \hat{V}_{j_2}^{i-1}\cap (\{p_{i-1}\}\cup Y_{i-1})$. Hence $r_v < \ell_w$ for any pair of vertices $v\in \hat{V}_{j_2}^{i-1}\cap (\{p_1,\dots, p_{i-2}\}\cup \bigcup_{j=1}^{i-2}Y_j)$, $w\in \hat{V}_{j_2}^{i-1}\cap (\{p_{i-1}\}\cup Y_{i-1})$ apart from the pair $p_{i-2}$, $p_{i-1}$. Since $Y_{i-1}\cap N_G(p_i)=\emptyset$ and $Y_i\subseteq \{v\in Y: i_v=i\}$ we further know that $r_v<\ell_w$ for every $v\in Y_{i-1}$, $w\in \{p_i\}\cup Y_i$. Combined we get that $r_v<\ell_{w}$ (and therefore $v$ is non-adjacent to $w$) for any $v\in \hat{V}_{j_2}^{i-1}\cap (\{p_1,\dots, p_{i-1}\}\cup \bigcup_{j=1}^{i-1}Y_j)$, $w\in \{p_i\}\cup Y_i$ apart from the pair $p_{i-1}$, $p_i$. Hence we argued that \ref{C4} holds for the partition $(V_1^i,\dots,V_6^k)$. Since both $p_{i-1}$ and $p_i$ are contained in $V_{j_2}^i$ the property \ref{C4} implies  that $\deg_{G[V_{j_2}^{i}]}(w)=\deg_{G[W_2]}(w)$ for every $w\in  Y_i$. Further,  \ref{C4} implies that $\deg_{G[V_{j_2}^{i}]}(w)=\deg_{G[\hat{V}_{j_2}^{i-1}]}(w)$ for every $w\in \hat{V}_{j_2}^{i-1}\setminus \{p_{i-1}\} $.
    Lastly, $\deg_{G[V_{j_2}^{i}]}(p_{i})=\deg_{G[W_2]}(p_{i})+1$ since $p_i$ is adjacent to $p_{i-1}$. Furthermore, $|\{p_1,\dots, p_{i}\}\cup \bigcup_{j=1}^{i}Y_j| $ is even if and only if $|Y_i|$ is odd since $|\{p_1,\dots, p_{i-1}\}\cup \bigcup_{j=1}^{i-1}Y_j| $ is odd. Since $\deg_{G[W_2]}(p_i)$ is even if and only if $|Y_i|$ is odd we get that $\deg_{G[V_{j_2}^{i}]}(p_{i})$ is odd if and only if $|\{p_1,\dots, p_{i}\}\cup \bigcup_{j=1}^{i}Y_j|$ is even as required. Therefore, \ref{C1} and \ref{C2} hold for $(V_1^i,\dots,V_6^i)$. Additionally, \ref{C3} is true by construction.\\

    Finally, since $|\{p_1,\dots, p_{k}\}\cup \bigcup_{j=1}^{k}Y_j|=|V(G)| $ is even  $(V_1^k,\dots,V_{6}^k)$ is an odd colouring of $G$ by \hyperref[C1]{$(C1)_k$}.
\end{proof}

\section{Conclusion}

We initiated the systematic study of odd colouring on graph classes. Motivated by Conjecture~\ref{conj:Scott}, we considered graph classes that do not contain large graphs from a given family as induced subgraphs. Put together, these results provide strong evidence that Conjecture~\ref{conj:Scott} is indeed correct. Answering it remains a major open problem, even for the specific case of bipartite graphs.

Several other interesting classes remain to consider, most notably line graphs and claw-free graphs. Note that odd colouring a line graph $L(G)$ corresponds to colouring the edges of $G$ in such a way that each colour class induces a bipartite graph where every vertex in one part of the bipartition has odd degree, and every vertex in the other colour part has even degree. This is not to be confused with the notion of odd $k$-edge colouring, which is a (not necessarily proper) edge colouring with at most $k$ colours such that each nonempty colour class induces a graph in which every vertex is of odd degree. It is known that all simple graphs can be odd 4-edge coloured, and every loopless multigraph can be odd 6-edge coloured (see e.g.,~\cite{Petrusevski18}). While (vertex) odd colouring line graphs is not directly related to odd edge colouring, this result leads us to believe that line graphs have bounded odd chromatic number.

Finally, determining whether Theorem~\ref{thm:oddChromaticModularWidth} can be extended to graphs of bounded rank-width remains open. We also believe that the bounds in Theorem~\ref{thm:oddChromaticInterval} and Corollary~\ref{cor:bounded-degree} are not tight and can be further improved. In particular, we believe that the following conjecture, first stated in~\cite{AashAGS23}, is true:

\begin{conjecture}
Every graph $G$ of even order has $\chiodd(G)\leq \Delta + 1$.
\end{conjecture}



\bibliography{oddColouring.bib}

\begin{thebibliography}{10}

\bibitem{AashAGS23}
Arman Aashtab, Saieed Akbari, Maryam Ghanbari, and Amitis Shidani.
\newblock Vertex partitioning of graphs into odd induced subgraphs.
\newblock {\em Discuss. Math. Graph Theory}, 43(2):385--399, 2023.
\newblock \href {https://doi.org/10.7151/dmgt.2371}
  {\path{doi:10.7151/dmgt.2371}}.

\bibitem{BelmonteS21}
R{\'{e}}my Belmonte and Ignasi Sau.
\newblock On the complexity of finding large odd induced subgraphs and odd
  colorings.
\newblock {\em Algorithmica}, 83(8):2351--2373, 2021.
\newblock \href {https://doi.org/10.1007/s00453-021-00830-x}
  {\path{doi:10.1007/s00453-021-00830-x}}.

\bibitem{Boll}
B\'{e}la Bollob\'{a}s.
\newblock Chromatic number, girth and maximal degree.
\newblock {\em Discrete Mathematics}, 24(3):311--314, 1978.

\bibitem{Caro94}
Yair Caro.
\newblock On induced subgraphs with odd degrees.
\newblock {\em Discret. Math.}, 132(1-3):23--28, 1994.
\newblock \href {https://doi.org/10.1016/0012-365X(92)00563-7}
  {\path{doi:10.1016/0012-365X(92)00563-7}}.

\bibitem{ChangD18}
Gerard~Jennhwa Chang and Guan{-}Huei Duh.
\newblock On the precise value of the strong chromatic index of a planar graph
  with a large girth.
\newblock {\em Discret. Appl. Math.}, 247:389--397, 2018.
\newblock \href {https://doi.org/10.1016/j.dam.2018.03.075}
  {\path{doi:10.1016/j.dam.2018.03.075}}.

\bibitem{Diestel}
Reinhard Diestel.
\newblock {\em Graph Theory, 4th Edition}, volume 173 of {\em Graduate texts in
  mathematics}.
\newblock Springer, 2012.

\bibitem{Erd1959}
Paul Erd\H{o}s.
\newblock Graph theory and probability.
\newblock {\em Canadian Journal of Mathematics}, 11:34--38, 1959.

\bibitem{FerberK22}
Asaf Ferber and Michael Krivelevich.
\newblock Every graph contains a linearly sized induced subgraph with all
  degrees odd.
\newblock {\em Advances in Mathematics}, 406:108534, 2022.
\newblock URL:
  \url{https://www.sciencedirect.com/science/article/pii/S0001870822003516},
  \href {https://doi.org/https://doi.org/10.1016/j.aim.2022.108534}
  {\path{doi:https://doi.org/10.1016/j.aim.2022.108534}}.

\bibitem{FominGLSZ19}
Fedor~V. Fomin, Petr~A. Golovach, Daniel Lokshtanov, Saket Saurabh, and Meirav
  Zehavi.
\newblock Clique-width {III:} hamiltonian cycle and the odd case of graph
  coloring.
\newblock {\em {ACM} Trans. Algorithms}, 15(1):9:1--9:27, 2019.
\newblock \href {https://doi.org/10.1145/3280824} {\path{doi:10.1145/3280824}}.

\bibitem{GanianHO13}
Robert Ganian, Petr Hlinen{\'{y}}, and Jan Obdrz{\'{a}}lek.
\newblock A unified approach to polynomial algorithms on graphs of bounded
  (bi-)rank-width.
\newblock {\em Eur. J. Comb.}, 34(3):680--701, 2013.
\newblock \href {https://doi.org/10.1016/j.ejc.2012.07.024}
  {\path{doi:10.1016/j.ejc.2012.07.024}}.

\bibitem{Lovasz93}
L{\'{a}}szl{\'{o}} Lov{\'{a}}sz.
\newblock {\em Combinatorial Problems and Exercises}.
\newblock North-Holland, 1993.

\bibitem{Petrusevski18}
Mirko Petrusevski.
\newblock Odd 4-edge-colorability of graphs.
\newblock {\em J. Graph Theory}, 87(4):460--474, 2018.
\newblock \href {https://doi.org/10.1002/jgt.22168}
  {\path{doi:10.1002/jgt.22168}}.

\bibitem{PetrS21}
Mirko Petrusevski and Riste Skrekovski.
\newblock Colorings with neighborhood parity condition, 2021.
\newblock URL: \url{https://arxiv.org/abs/2112.13710}, \href
  {https://doi.org/10.48550/ARXIV.2112.13710}
  {\path{doi:10.48550/ARXIV.2112.13710}}.

\bibitem{Scott01}
Alex~D. Scott.
\newblock On induced subgraphs with all degrees odd.
\newblock {\em Graphs Comb.}, 17(3):539--553, 2001.
\newblock \href {https://doi.org/10.1007/s003730170028}
  {\path{doi:10.1007/s003730170028}}.

\end{thebibliography}

\appendix

\end{document}